\documentclass[11pt]{article}
\usepackage[utf8]{inputenc}
\usepackage[a4paper,top=3cm,bottom=3cm,left=2cm,right=2cm,bindingoffset=5mm]{geometry}
\usepackage{amsmath,amsfonts,amssymb,amsthm}
\usepackage[dvips]{graphicx}
\usepackage{algorithm}
\usepackage{algpseudocode}
\usepackage{bbold}
\usepackage{appendix}
\usepackage{mathtools}
\usepackage{doi}
\mathtoolsset{showonlyrefs}
\usepackage[dvipsnames]{xcolor}
\title{ \textbf{ Parallel Load Balancing on Constrained  Client-Server Topologies }}

\author{Andrea Clementi \\ {\small{}Università di Roma Tor Vergata}\\
    {\small{} Rome, Italy}\\
    {\small{}clementi@mat.uniroma2.it} 
    \and Emanuele Natale\\
    {\small{}Universit\'e C\^ote d'Azur, CNRS, INRIA}\\
    {\small{} Sophia Antipolis, France}\\
    {\small{}natale@unice.fr} 
    \and Isabella Ziccardi \\ 
    {\small{}Università dell'Aquila,}\\
    {\small{} L'Aquila, Italy}\\
    {\small{}isabella.ziccardi@graduate.univaq.it}
     }

\date{}

\begin{document}
% Isa's todo
\newcommand{\isam}[1]{\todo[color=yellow]{I: #1}}
\newcommand{\isa}[1]{\todo[inline,color=yellow]{I: #1}}
% andy's todo
\newcommand{\andym}[1]{\todo[color=purple]{A: #1}}
\newcommand{\andy}[1]{\todo[inline,color=purple]{A: #1}}
% ema's todo
\newcommand{\emam}[1]{\todo[color=pink]{E: #1}}
\newcommand{\ema}[1]{\todo[inline,color=pink]{E: #1}}

\newtheorem{definition}{Definition}
\newtheorem{prop}[definition]{Proposition}
\newtheorem{lemma}[definition]{Lemma}

\newtheorem{fact}[definition]{Fact}
\newtheorem{claim}[definition]{Claim}
\newtheorem{theorem}[definition]{Theorem}
\newtheorem{cor}[definition]{Corollary}
\newtheorem{obs}[definition]{Observation}
\newtheorem{example}[definition]{Example}
\newtheorem{exercise}[definition]{Exercise}
\newtheorem{conj}[definition]{Conjecture}
%\newfloat{pseudocode}{thb}{pseudo}
%\floatname{pseudocode}{Pseudocode}

\newcommand{\bigO}{\mathcal{O}}
\newcommand{\Prob}[2]{\mathbf{P}_{#1} \left( #2 \right)}
\newcommand{\Prc}[1]{\mathbf{Pr} \left( #1 \right)}
\newcommand{\Expec}[2]{\mathbf{E}_{#1} \left[ #2 \right]}
\newcommand{\Expcc}[1]{\mathbf{E} \left[ #1 \right]}
\newcommand{\Var}[2]{\mathrm{Var}_{#1} \left( #2 \right)}
\newcommand{\Cov}[2]{\mathrm{Cov}_{#1} \left( #2 \right)}
\newcommand{\skproof}{\noindent\textit{Sketch of Proof. }}
\newcommand{\ideaproof}{\noindent\textit{Idea of the proof. }}

\newcommand{\polylog}[1]{\mathrm{polylog}\left(#1\right)}
\newcommand{\poly}{ {\mathrm{poly}}}

% ANDY'S COMMANDS
\newcommand{\sS}{\mathcal{S}}
\newcommand{\sC}{\mathcal{C}}
\newcommand{\local}{{\sc{local}}}
\newcommand{\gossip}{{\sc{gossip}}}
\newcommand{\algo}{{\sc ALG}}
\newcommand{\ALG}{{\sc raes}}
\newcommand{\raes}{{\sc raes}} 
\newcommand{\saer}{{\sc saer}}
 \newcommand{\degC}{\Delta_{\text{min}}(\mathcal{C})}
 
 \newcommand{\degS}{\Delta_{\text{max}}(\mathcal{S})}
\newcommand{\din}{ {\mathrm{d_u^{in}}}}
\newcommand{\grado}{\ensuremath{d}}
\newcommand{\bC}{\mathbf{c}}
\newcommand{\bx}{\mathbf{x}}
\newcommand{\by}{\mathbf{y}}
\newcommand{\sM}{\mathcal{M}}
\newcommand{\SST}{\mathrm{SS}}

%ISA'S COMMAND
\newcommand{\alphaminC}{\alpha_{\text{min}}(\mathcal{C})}
\newcommand{\alphamaxC}{\alpha_{\text{max}}(\mathcal{C})}
\newcommand{\alphaminS}{\alpha_{\text{min}}(\mathcal{S})}
\newcommand{\alphamaxS}{\alpha_{\text{max}}(\mathcal{S})}
\newcommand{\degmaxC}{\Delta_{\text{max}}(\mathcal{C})}
\newcommand{\degminS}{\Delta_{\text{min}}(\mathcal{S})}
\newcommand{\hmax}{h_{\text{max}}}
\newcommand{\hmin}{h_{\text{min}}}
 \newcommand{\dout}{ {\mathrm{d_v^{out}}} }
 \newcommand{\asout}{\mathrm{A_S^{out}}}
 \newcommand{\Cost}{\mathrm{Cost}}
 \newcommand{\Rappr}{\mathrm{Rappr}}
 \newcommand{\alphamin}{\mathrm{\alpha_{min}}}
 \newcommand{\alphamax}{\mathrm{\alpha_{max}}}
 \newcommand{\Deltamin}{\mathrm{\Delta_{min}}}
 \newcommand{\Deltamax}{\mathrm{\Delta_{max}}}
 \newcommand{\cost}{ {\mathrm{Cost} }}

\newcommand{\bin}{\mathbf{Cost}}
\newcommand{\dest}{\mathrm{Dest}}
\newcommand{\sRJ}{\mathrm{Rej}} 
%\newcommand{\state}{\ensuremath{D}}

% LUCA'S COMMANDS
\newcommand{\nreq}{\ensuremath{d}} % This is d
\newcommand{\eventC}{\ensuremath{\mathcal{C}}}
\newcommand{\eventE}{\ensuremath{\mathcal{E}}}
\newcommand{\rc}{\ensuremath{rc}}
\newcommand{\rss}{\ensuremath{rss}}
\newcommand{\vol}{\ensuremath{vol}}
\newcommand{\out}{\ensuremath{\delta}}
\newcommand{\DEC}{\sc{Dec}}

% Fraction of requests 
\newcommand{\freq}{\ensuremath{\epsilon}} 
% Subset of critical nodes
\newcommand{\crit}{\ensuremath{C}}

% FRANCESCO'S COMMANDS
\renewcommand{\leq}{\leqslant}
\renewcommand{\le}{\leqslant}
\renewcommand{\geq}{\geqslant}
\renewcommand{\ge}{\geqslant}
\renewcommand{\epsilon}{\varepsilon}
\newcommand{\bX}{\mathbf{X}}
\newcommand{\bY}{\mathbf{Y}}

%pacchetto per mostrare solo le etichette utilizzate

\maketitle
\doi{10.1145/3350755.3400232}

\begin{abstract} 
 We study   parallel \emph{Load Balancing}  protocols for   a client-server distributed model defined as follows.
      There is a set $\sC$ of  $n$  clients and  a set $\sS$ of $n$  servers where  each client  has
      (at most) a  constant  number $d \geq 1$ of  requests  that  must be  assigned to some  server.  The client set and the server one are connected to each other via a fixed  bipartite graph: the requests of  client $v$ can only be  sent to the servers in its neighborhood  $N(v)$. The  goal is to  assign every client request  so as to minimize the maximum load of the servers. 
       
 In  this  setting, efficient parallel protocols are available only for   dense topolgies. In particular, a simple symmetric, non-adaptive protocol achieving constant maximum load
 has been recently introduced by Becchetti et al \cite{BCNPT18}
 for regular dense bipartite graphs.
The  parallel completion time is $\bigO(\log n)$ and the overall work is    $\bigO(n)$, w.h.p.

 Motivated by proximity constraints arising in some client-server systems, we devise  a simple  variant of Becchetti et al's protocol \cite{BCNPT18}   and we  analyse it  over  almost-regular  bipartite  graphs where nodes may  have  neighborhoods of small size.  
 In detail,   we prove that, w.h.p.,  this new version  has  a  cost equivalent to  that of   Becchetti et al's protocol (in terms of maximum load, completion time,  and work complexity, respectively) on every almost-regular bipartite graph with  degree $\Omega(\log^2n)$.

 Our analysis significantly departs from that  in \cite{BCNPT18} for the original protocol and requires to cope with   non-trivial    stochastic-dependence issues  on the random  choices   of the algorithmic process which are due to the  worst-case, sparse topology of the underlying graph.
    
\end{abstract}

%\keywords{Parallel Balanced Allocations,  Balls-into-Bins Processes, Randomized Algorithms}



\maketitle
% \clearpage

\section{Introduction} 

\subsection{The Framework and  our Algorithmic Goal} \label{ssec:introframe}
We study  parallel   \emph{Load-Balancing} allocation in client-server distributed systems. We have a  client-server bipartite graph $G(V=(\sC,\sS),E)$ where:       $\sC$ is the set of \emph{clients}, each one having a   number of   requests which is bounded by some constant $d \geq 1$;  $\sS$ is the set of \emph{servers};   the edge set $E$ 
represents the client-server assignments which are considered admissible because of proximity constraints (a client can send a request only to the servers in its neighborhood). 

The algorithmic goal of the entities is to assign the requests in parallel so as to minimize the  maximum server load\footnote{the load of a server is the overall number of requests which have been assigned to it. }.

To analyze the performance
of the proposed  protocol for the above distributed task, we adopt the standard synchronous distributed model introduced for parallel \emph{balls-into-bins} processes by Micah et al in \cite{ACMR98}: here,   clients and servers 
are autonomous computing entities that can  exchange information (only)  over the edges of   $G$.
   Micah et al introduce  the class of symmetric, non-adaptive protocols  and show several tight bounds   on the trade-offs between the  maximum load and the complexity (i.e. \emph{completion time} and \emph{work complexity}\footnote{The work complexity is the  overall number of exchanged messages performed by the protocol.}) of the proposed solutions. Informally, a protocol is said to be \emph{symmetric} if the entities   are anonymous,  so  all the clients (servers) act in the same way and, moreover, all possible request destinations  are chosen independently and uniformly at random.  The protocol is said to be \emph{non-adaptive} if each 
client   restricts itself to a fixed number of (possibly random)
candidate servers in its neighborhood before communication starts.
Symmetric, non-adaptive protocols have the practical merits to be easy to implement and more flexible \cite{ACMR98}. Such solutions have interesting applications in Computer Science, such as load balancing in communication networks, request scheduling and hashing \cite{AAFPW97,AHKL05,BLSZ14,R02}.

We notice that efficient  symmetric, non-adaptive
  protocols  are known (only) for  dense  regular     bipartite graphs   and   almost-tight lower bounds are known for this important class of parallel protocols \cite{ACMR98,BCNPT18,LPY19}
  (see also  Subsection \ref{ssec:previous} for a short description of such results).
  
   The main goal of this paper does not consist of    improving previous solutions with respect to specific complexity measures. Rather, still aiming at    efficient  solutions that achieve bounded  
   maximum load\footnote{According to our parameter setting, the maximum load is clearly at least $d$ and we aim at keeping  an $O(d)$ bound for it.},  we focus on symmetric, non-adaptive  Load-Balancing protocols that work  over  restricted, non dense graph topologies.
   This natural extension of previous work  is inspired by possible  network applications where: i)  based on previous experiences, a client (a server)  may decide to send (accept) the  requests only to (from) a fixed subset of \emph{trusted} servers (clients) and/or ii)
 clients and servers   are placed over a metric space so that 
 only  non-random  client-servers interactions turn out to be
  feasible because of proximity constraints. Such possible scenarios  motivated previous important studies on   sequential Load-Balancing algorithms \cite{BBFN10,G08,KP06}. 
   To the best of our knowledge,  
    efficient solutions for non-dense graphs are in fact available only for the  classic sequential 
    model. Here,     
each client request is scheduled once at  time  so that, for instance, the  well-known \emph{best-of-$k$-choices}  strategy \cite{ABKU94} can be applied:       the loads of   the servers are updated at each   assignment and  the new considered request is assigned to   a server 
that has the current minimal load out of  $k$   
servers chosen independently and  uniformly at random \cite{BBFN10,G08,KP06}.

 As for the parallel distributed model we adopt in this paper,    in \cite{BCNPT18} Becchetti et al propose a symmetric, non-adaptive algorithm, named \raes\   (for
\emph{Request a link, then Accept if Enough Space}), which is based on the well-known \emph{threshold} criterion \cite{ACMR98}. Informally, \raes\ works 
 in rounds, each consisting of two phases. Initially, each client has $d=\Theta(1)$ balls\footnote{The terms \emph{ball} and \emph{request} will be used interchangeably.}. In the first phase of each round, if  client $u$  has $d' \geq 1$ \emph{alive} balls  (i.e. to be still accepted by some server), $u$  selects  $d'$ servers independently and uniformly at random (with replacement) from $N(u)$.  It then submits each of the $d'$ balls
to each selected client. In the second phase of the
round, each server accepts all requests received in the first phase of the
current round, unless doing so would cause it to exceed the limit of $cd$
accepted balls, where the parameter $c$ is a suitable large constant; if this is the case, the server is said to be \emph{saturated} and rejects all requests it received in the first phase of the current round.  The algorithm completes when every client has     no further balls to be  submitted. 

Observe that servers only give back Boolean answers to the clients requests and, moreover,  if the algorithm terminates, the maximum load of the  servers will be at most $cd$.
Becchetti et al prove\footnote{Not related to our context, the main result in Becchetti et al    shows  that \raes\ can be used to  construct a bounded-degree expander subgraph of     
  $G$, w.h.p.} that, over any $\Delta$-regular bipartite graph with $\Delta = \Omega(n)$, \raes\ terminates within $O(\log n)$  rounds and the total work  is $\Theta(n)$, \emph{with high probability}\footnote{As usual, we say that an event $E$ holds with high probability  if   a constant $\gamma > 0$ exists such that $\mathbf{P}(E) \ge 1-n^{-\gamma}$.} (for short, \emph{w.h.p.}).

\subsection{Our Contribution}

We consider a variant of \ALG, called \saer\  (\emph{Stop Accepting if Exceeding Requests}) that works like \raes\  with the   exception that, whenever a server $v$, in the second phase of a given round, gets an overall load larger than $cd$, then $v$ rejects all requests arrived in the first phase of the current round and it becomes \emph{burned}. Once a server gets burned,  it  will never accept
  any  request for all successive rounds (see Algorithm 1 in Subsection \ref{sec:prely}). 
  
  Similarly to   \raes,   if  this new version    terminates, then     each server will have load at most $cd$ and, hence, the main technical issue is to provide a
  bound (if any) on the number of rounds required by \saer\ to let every client ball
  assigned to some server.
  
  We prove that, for any almost-regular bipartite graph $G(V=(\sC,\sS),E)$
  of   degree $\Delta = \Omega(\log^2n)$ (recall that $|\sC|=|\sS| = n$), it is possible to choose a sufficiently large constant   $c \geq 0$, such that, for any constant request number $d$,  the protocol \saer\ terminates within $O(\log n)$ rounds and  requires  $\Theta(n)$ work,
  w.h.p.
  
  Informally, for \emph{almost-regular} bipartite graphs we mean bipartite graphs where the ratio between the minimum degree of the client set and the maximum degree of the server set is bounded by an arbitrary  positive   constant (see Theorem \ref{thm:SAER-terminates} for a formal definition). Observe that this notion of almost regularity  allows a certain variance of the degrees of entities of the same type: just as a (``non-extremal'') example, we may consider a bipartite graph where:  most of the  clients have (minimal) degree $\Theta(\log^2n)$, while few  of them have degree $\Theta(\sqrt n)$;  most of the servers have (maximal) degree $\Theta(\log^2n)$, while  few of them have degree $o(\log n)$.

\medskip
\noindent{\textbf{Algorithm Analysis: An Overview.}}
In the case of dense graphs, the key-fact exploited by      Becchetti et al's analysis of the \raes\ algorithm  \cite{BCNPT18} is the following. 
Since each client
has  $\Theta(n)$   servers in its neighborhood,
it is possible to fix a sufficiently large 
constant $c$, such that, at every round, the fraction  of non-burned\footnote{Recall that a   server is burned  at round $t$ if its load is larger than  $cd$.}  servers
in the  neighborhood $N(v)$ of every client $v$ is always at least $1/2$.
Thanks to  a basic counting argument, this fact holds \emph{deterministically} and \emph{independently} of  the previous load configurations yielded 
by   the process. So,   every alive client request 
has probability at least $1/2$ to be accepted at each  round:   this allows to get a logarithmic completion time of  \raes\ on dense graphs.

In the case of non-dense  graphs (i.e. for node degree $o(n)$), the key property above does not hold deterministically: the fraction of non-burned  servers in a fixed neighborhood  is a random variable that can even take value $1$  and, very importantly, it  depends on the graph topology and on the random choices performed by the nodes during the previous rounds.
This scenario  makes the analysis considerably harder than that of the dense case.
To cope with the above issues, for an arbitrary  client $v$, we look at its server neighborhood $N(v)$ and we   establish   a clean recursive formula that describes the expected decreasing rate of the overall number $r_t(N(v))$ of requests that the neighborhood of $v$ receives at time $t$. This expectation is derived  for round $t$ by conditioning on the sequence of the maximum fractions of burned servers in any client's neighborhood produced  by the algorithmic process at rounds $1, 2, \ldots, t-1$.
It turns out that, for a sufficiently large $c$, the  conditional  expected  decreasing rate of $r_t(N(v))$ is exponential.
Then, using a coupling argument, we derive a concentration bound  for 
this    rate that holds as long as the conditional expectation of  $r_t(N(v))$ keeps  of magnitude $\Omega(\log n)$. To complete our argument, we  
consider a further (and final) stage of the process\footnote{Notice that this 
stage is   only in our analysis and not on the protocol, the latter being symmetric and non-adaptive.} that starts when $r_t(N(v)) = O(\log n)$: here,    we do not look anymore at the decreasing rate of $r_t(N(v))$, rather we show that, w.h.p,  the  fraction of burned servers in $N(v)$ can  increase, along    a time window of length $\Theta(\log n)$,  by  
an overall additive  factor  of magnitude at most  $O(1/c)$. Thanks to this fact, we can then show that the   $O(\log n)$ requests that survived the first stage  have high chances to be assigned during this last stage if the  latter lasts   $\Theta(\log n)$ additional rounds.

\medskip
\noindent{\textbf{Remark.}} We observe that, 
  while  the notion of burned server plays a crucial role  in our   analysis of  \saer,  this notion is stronger than that of saturated servers adopted by the original protocol \raes. Hence, 
  our bounds on the termination time  and the work complexity of the \saer\ protocol can be easily extended to the original Becchetti et al's protocol \raes. 

%%%%%%%%%%
 
\subsection{Previous Work} \label{ssec:previous}
Load-Balance  algorithms     
have been the subject of a long and extremely active line of research with important applications in several topics of Computer Science such as hashing, PRAM simulation, scheduling, and load balancing.
  A well-established  and effective way to model such problems is by using the classic \emph{balls-into-bins} processes. In such processes, there are typically $m$ balls that must be assigned to $n$ bins.
  In what follows, we use this framework to 
shortly describe   those previous  results which are more related to the setting of this work.

\smallskip
\noindent
\textbf{Sequential Algorithms on the Complete Bipartite Graph.} It is well-known that if   $n$ balls are thrown independently and uniformly at random into $n$ bins,  the maximum load  of a bin  
   is bounded by   $\Theta(\log n/\log\log n)$, w.h.p (see for instance \cite{MU17}).   
    Azar et al. \cite{ABKU94} proved the following breakthrough result.   Assume the balls are assigned  sequentially, one at a time and, for  each ball,   $k \geq 1$ bins are chosen independently
and uniformly at random, and   the ball is assigned to the least full bin (with ties   broken
arbitrarily). This greedy strategy is also known as ``best of $k$ choices". Then,   they prove that the final  maximum load is  
$\Theta(\log\log n/\log k + 1)$, w.h.p. A similar result  was also derived in a different version of the model by Karp et al in  \cite{KLMF92}. Berenbrink et al  extended the analysis of the Greedy algorithm for the heavily-loaded case $m>>n$ \cite{BCSV06}.
Then, several versions of this sequential algorithm have been studied 
   by considering, for instance,    non-uniform choices in  the assignment process \cite{BCM04,V03,W07}.  
   Moreover, several works addressed weighted balls \cite{BFHM08,BMFS97,KMS03}, while the  case of heterogeneous bins was studied in   \cite{W07} . Recently, balls-into-bins processes have  also been analyzed  over   game theoretic frameworks \cite{BFGGHM06,KT09}.
   
   \smallskip
\noindent
\textbf{Sequential Algorithms on Restricted Bipartite Graphs.} Sequential algorithms for restricted balls-bins (i.e. client-server)  topologies  have been considered in \cite{BBFN12,G08,KP06}: here,   each ball    $u= 1, \ldots , m$ comes with its admissible cluster of bins and decides its strategy according to the  current loads in its cluster determined by the choices of the previous balls $u' < u$.   
 In this setting,       Kenthapadi and Panigrahy \cite{KP06}  analyse the well-known sequential Greedy algorithm \cite{ABKU94}: each client $u$, in turn,  chooses a pair of servers uniformly at random from $N(u)$ and  assigns the ball to  the  server having the current minimum load. They prove that, if the size  of every $N(u)$ is at least $n^{\Theta(1/\log\log n)}$, then the Greedy algorithm achieves   maximum load $\Theta(\log\log n)$, w.h.p.
 In \cite{G08}, Godfrey analyzed the sequential Greedy algorithm on the input model where a \emph{random} cluster of  servers $N(u)$ is assigned  to each client $u$ before the algorithm starts. In more detail, each client $u$  places its  ball     in a  uniform-random server  among those  in $N(u)$ with the current fewest number of balls.
He proves that, if the random subsets $N(u)$ are chosen according to any fixed  \emph{almost-uniform}  distribution over the server set $\sS$ and the subsets $N(u)$ have size $\Omega(\log n)$, then the Greedy algorithm achieves  optimal  maximum load, w.h.p.. The overall work is $\Theta(n  \Deltamax(\sC))$, where  $\Deltamax(\sC) = \max\{|N(u)| \, : \, u \in \sC \}$.  Further bounds are determined when the overall number $m$ of balls is smaller than the size of the server set $|\sS| = n$.
Berenbrink et al \cite{BBFN12} consider the   sequential framework adopted in \cite{G08} and improve   the analysis of the greedy algorithm along different directions.  In detail, they consider weaker notions of almost-uniform distributions for the random server clusters assigned to the clients and, moreover,  they also consider  an input framework formed by  deterministic, worst-case server clusters of   size $\beta \log n$. In the case where the overall number of balls is $n < \alpha m$,  with  any  $\alpha < 1/12$ and $\beta >18$, they show that a suitable version of the sequential greedy algorithm achieves maximum load 1, w.h.p.
Notice that   the Greedy algorithm adopted in \cite{KP06,G08} does   require every  server to give information to  their clients about  its current load: in some  applications, this feature of the algorithm  might yield  critical issues in terms of privacy and security of the involved entities \cite{GP13,ZSZF10}. On the other hand, we notice that, the simple threshold approach adopted  by both Becchetti et al's Algorithm \saer\ and 
our version  \raes\ can be implemented in a fully decentralized fashion so that the clients cannot get a good approximation about the current load of the  servers (see also the remark after Algorithm 1 in Subsection \ref{ssec:theprot}).

\smallskip
\noindent
\textbf{Parallel Algorithms on the Complete Bipartite Graph.}
Inspired by applications arising from parallel distributed systems, a rich and active research  has been focused  on 
computational entities which are able to communicate each other (with some constraints that depend on the specific version of the model). Then,
protocols operate in synchronous rounds, in each of which balls and bins exchange messages once. In \cite{ACMR98}, Micah et al consider 
some non-adaptive symmetric protocols and analyze their performances in terms of maximum load, number of rounds, and
message complexity. For instance, they introduce a parallelization of the Greedy algorithm \cite{ABKU94} and show that for any constant  number $r$ rounds and     for any constant number of random choices $k$, it  achieves maximum load $\bigO((\log n/\log\log n)^{1/r})$, w.h.p. They also give a more complex Greedy algorithm that works in  $\log \log n /\log k + 2k  + \bigO(1)$ rounds and achieves $\log \log n /\log k + 2k  + O(1)$ maximum load, w.h.p. 
Interestingly enough, they   prove that the above performance trade-offs  are essentially optimal for the restricted class of non-adaptive, symmetric algorithms.
This class also includes the \emph{Threshold} algorithms where, informally speaking, at every round, every bin that receives more than a fixed threshold $T$ of balls, 
the excess balls are re-thrown in the next round (such rejected balls can be chosen in an arbitrary ``fair'' way). 
Parallel algorithms based on the threshold approach have been introduced by Lenzen et al in \cite{LPY19} for the  heavily-loaded case, i.e. when $m >>n$.  Finally, we mention   some adaptive and/or non-symmetric protocols on the complete  graph  that have 
 been presented in   recent work   (e.g. \cite{BFLS,BKSS13,LPY19,LW11})   that achieve significantly 
 better performances than  symmetric and/or non-adaptive ones  \cite{ACMR98}.
 Such     strategies are rather complex and  so    their setting is
   far from  the   aim of this paper (as discussed in the previous subsection, this being the analysis of   basic, non-adaptive symmetric protocols over
 restricted client-server topologies).

\smallskip
\noindent
\textbf{Parallel  Algorithms  on Restricted  Bipartite Graphs.}
The only rigorous analysis of parallel protocols
for restricted client-server topologies we are aware of  is  that in \cite{BCNPT18} by Becchetti et al for the \raes\ protocol which has been discussed in the previous part of this introduction.

 %%%%%%%%%%%%%%%%%%

\section{The \saer\ Protocol and the Main Theorem}

\subsection{Preliminaries} \label{sec:prely}
In the  \emph{Load-Balancing} problem we have
a system formed by a client-server bipartite graph $G(V=(\sC,\sS),E)$ where: 
the   subset $\sC = \{v_1,\ldots,v_n\}$
represents the set of \emph{clients}, the   subset
$\sS= \{u_1, \ldots , u_n\}$ represents the set of \emph{servers}, and the edge set $E$ determines,
for each client $v$, the subset $N(v)$ of servers the client  $v$ can make a \emph{request} to (i.e. it can send a \emph{ball}\footnote{Recall that the terms \emph{ball} and \emph{request} will be used interchangeably.}).
At the beginning, each client has at most $d$   balls 
where $d \geq 1$ is an arbitrary constant (w.r.t. $n$) that, in the sequel, we call \emph{request number}, and the goal  is to design a parallel distributed protocol   that assigns each ball of every client $v \in \sC$ to
one server in $N(v)$.

According to previous work \cite{ACMR98,LW11}, we study the Load-Balancing problem over the  fully-decentralized  computational model $\sM$ where bi-directional communications   take place only along the edges in $E$, in synchronous rounds. Moreover,  clients may only send the ball IDs\footnote{It suffices that each     client  keeps    a local labeling of its ball set. }, while servers  may only answer each ball request with one bit: accept/reject.
There is no global labeling of the nodes of $G$: each node $v$ just keeps a local labeling of its links.

We analyze the cost of the proposed solution with respect to two complexity measures:  the \emph{completion  time}  which is  defined as the number of rounds required by the protocol to successfully assign  all the client balls to the servers;   the (overall) \emph{work} which is defined as the overall number of exchanged messages among the nodes of the network during  the protocol's execution. 

For any node (client or server) $w \in V = (\sC,\sS)$, we denote its degree in $G$ as $\Delta_w$, i.e. $\Delta_w=|N(w)|$ and we define
\[ \degC \, = \, \min\{\Delta_v \, : \ v \in \sC \} \, \mbox{ and } \, \degS \, = \, \max\{\Delta_u\, : \ u \in \sS \} \, . \]

 \subsection{ A Simple  Protocol for Load Balancing} \label{ssec:theprot}
  As described in the introduction, 
 the protocol we propose in this paper  is a variant of the protocol \raes\ introduced in \cite{BCNPT18} and it is based on a simple, non-adaptive threshold  criterion  the servers use to accept or rejects the incoming balls. The protocol  is organized in rounds and, in turn, each round consists of two phases.  
 For the sake of readability, we consider the case where every client has exactly $d$ balls, where the request number $d$ is an arbitrary fixed constant:   the analysis of the general case ($\leq d$) is in fact similar.

\begin{algorithm}[H]
\caption{Protocol \saer($c,d$)}\label{alg:raes}
	\begin{algorithmic}[1]
		\State {\sc Phase 1:}\Comment{$\dout$: current number of the   balls  of  $v \in \sC$ which have been accepted by some server}		
			\For {$v\in \sC$}
				\State $v$ picks $d-\dout$ neighbors in $G$, independently and  uniformly at random (with replacement)
				\State $v$ submits a ball request to each of them
			\EndFor
			\State {\sc Phase 2:}
			\Comment{$\din$: current number of the balls   accepted by   $u \in \sS$}	
			\For {$u\in \sS$}
			    \If {$u$ is {\sc burned}}
			        \State $u$ rejects all the balls received in Phase 1 of the current round 
			    \Else 
				\If {$u$ received $> cd$  balls since the start of the process}
					\State $u$ rejects all the  balls received in Phase 1 of the current round and becomes   {\sc burned} 
				\Else
					\State $u$  accepts all of the balls received in Phase 1 and updates  $\din$
				\EndIf
				\EndIf
			
			\EndFor
			\For {$v\in \sC$}
			 \State{$v$ updates its value $\dout$}
			        \If {$\dout = 0$} \State  $v$ gets into the final state {\sc done} and terminates.
			        \EndIf
			\EndFor
			 	\end{algorithmic}
\end{algorithm}

\smallskip
\noindent 
\textbf{Remarks.} Some simple facts   easily follow from  the  protocol description above.
(i) The   protocol completes  at round $T \geq 1$ if and only  if every client has successfully  placed all its $d$ balls within round $T$. If this happens, then the maximum load of the servers is clearly  bounded by $cd$. The main technical question is thus  to provide   bounds in concentration  on the completion time of the protocol and on its performed work. This issue will be the subject of    the next section. \\ (ii) As for the decentralized implementation of \saer($c,d$), we observe that the knowledge of the parameter $c$ (which, in turn,  depends on the    degree of the underlying almost-regular bipartite graph - see Theorem \ref{thm:SAER-terminates} in the next subsection)  is required only by the servers while clients need no  knowledge of   global parameters. Interestingly enough, this fact implies that, for reasons of security and/or privacy,   the servers may suitably choose  $c$ so that the clients cannot get any good approximation of their current load.

\subsection{ Performance Analysis of \saer} \label{sec:analysis}
Using   the definition of    client-server 
 bipartite  graphs and    that of   Protocol \saer$(c)$ we gave in the previous subsections, we can state our  main technical contribution as follows. 
\begin{theorem}
\label{thm:SAER-terminates} 
 Let $\eta$ and $\rho \geq 1$ be two arbitrary   constants in $\mathbb{R^+}$ and let $d$ be an arbitrary constant in $\mathbb{N}$.    Let   $G((\sC,\sS),E)$ be any  bipartite graph such that 
$\Deltamin(\sC) \geq \eta  \log^2 n$ and  $\Deltamax(\sS)/ \Deltamin(\sC) \leq  \rho$. Consider the Load-Balancing problem on $G$   with  request number  $d$.     Then,  there is a sufficiently large   constant  $c>0$,\footnote{Our analysis will show that  the value of $c$  depends (only) on the constants $\eta$ and $\rho$.} such that   \saer($c,d$)\  has completion time $\mathcal{O}\left(\log n\right)$ and its   work  is $\Theta(n)$,  w.h.p.
\end{theorem}

Since the notion of burned servers adopted  in \saer \ (see Definition \ref{def:S_t(v)}) is stronger than the notion of saturated  servers adopted in the Becchetti et al's protocol
 \raes($c,d$)\  \cite{BCNPT18} (see Section \ref{ssec:introframe}), it is easy to verify that the number of accepted client requests at every round of the \raes \ process   is stochastically dominated by the same random variable in the \saer\ process. This fact implies the following result.

\begin{cor} \label{cor:raes-temination}
Under the same hypothesis of Theorem \ref{thm:SAER-terminates}, 
there is a sufficiently large constant $c$ such that
    \raes($c,d$)\  has completion time $\mathcal{O}\left(\log n\right)$ and its   work  is $\Theta(n)$,  w.h.p.
\end{cor}

  A simple counting argument implies that
  $\Deltamin(\sC) \leq \Deltamax(\sS)$ for any  bipartite graph while  Theorem \ref{thm:SAER-terminates} requires the ``almost-regularity'' hypothesis $\Deltamax(\sS)/ \Deltamin(\sC) = \Theta(1)$.   On the other hand, we emphasize that this      condition allows a relative-large variance of the node degree. 
  For instance, the theorem holds for  a  topology where:  
     the minimum client degree and the maximum server degree   are $\Theta(\log^2 n)$, some clients have degree $\Theta(\sqrt{n})$, and some servers have (minimal) degree $\Theta(1)$.

In the next section, we  will prove   Theorem \ref{thm:SAER-terminates}
in the case of $\Delta$-regular bipartite graphs then, in Appendix \ref{ssec:irregular}, 
we will show how to adapt the analysis for the
more general graphs considered in the theorem. We decided to  distinguish 
the two cases above for the sake of readability:    the regular case essentially includes all the main technical ideas of our analysis while allowing a  much simpler notation.

\section{Proof of Theorem 1: The Regular Case}\label{subsec:thm_1}

We prove here Theorem \ref{thm:SAER-terminates} for an arbitrary   
   $\Delta$-regular bipartite graph $G(V=(\sC,\sS),E)$  where $\Delta = \Delta(n)$ is any function in $\Omega(\log^2 n)$.   Since the protocol \saer\ makes a crucial  use    of
\emph{burned} servers, in what follows, we  define  this notion    and  some important  random variables of the algorithmic process which are  related to it.
For each round $t \geq 1$ and each server $u \in \sS$, let $r_t(u)$ be the random variable indicating the number of balls that server $u$ receives at time $t$.

\begin{definition} \label{def:S_t(v)}
A server  $u \in \sS$  is \emph{burned}
at round $t$ if $\sum_{i=1}^{t}r_{i}\left(u\right)\geq cd$. 
Moreover,  for any client $v \in \sC$, 
define $S_{t}\left(v\right)$
as the fraction of burned servers in the neighborhood of $v$ at time
$t$, i.e., 
\begin{equation*}
S_{t}\left(v\right)=\frac{\left|\left\{ u \in \sS :u\in N\left(v\right)\text{ and }\sum_{i=1}^t r_t(u) \geq cd\right\} \right|}{\Delta} \, .
%\label{def:s_burned}
\end{equation*}
We also  define  $S_{t}$ as the maximum fraction of burned nodes in any client's neighborhood at round $t$, i.e.,
     $S_{t} \ = \ \max_{v \in \sC}S_{t}\left(v\right)$.
\end{definition}

 The proof of Theorem  \ref{thm:SAER-terminates} relies on    the following result.
 
\begin{lemma} \label{lem:fraction_burned}  Let $\Delta \geq \eta \log^2 n$ for an arbitrary constant  $\eta$ in $\mathbb{R^+}$ and let $d$   be  an  arbitrary   constant in $\mathbb{N}$. 
Then, for any  $c \geq \max(32, 288/(d \eta))$ and for a sufficiently large $n$, with probability at least $1-1/n^{2}$,  it holds that for all rounds $t \leq 3 \log n$ the fraction of burned nodes satisfies

\begin{equation}  \label{eq:lemma_burned}
  S_t \leq \frac 12 \, .
  \end{equation}
\end{lemma}

We observe that the bound on the completion time stated in Theorem \ref{thm:SAER-terminates} for the regular case with
$\Delta = \Omega( \log^2n)$ is a simple consequence of the above lemma. 
Indeed,  consider any fixed ball of a client $v \in \sC$. By choosing\footnote{Since $d \geq 1$, the suitable value for $c$ can be fixed by the servers by looking only at $\eta$. We also remark our analysis does not optimize several aspects   such as the bound on $c$ and its relation with $\eta$. } the parameter $c$ as indicated  by Lemma \ref{lem:fraction_burned},  \eqref{eq:lemma_burned} implies that  the probability   the ball is  not  accepted for all rounds $t \leq 3 \log n$, conditioning  on  the bound given in Lemma \ref{lem:fraction_burned}, is $\left(1/2\right)^{3 \log n}= (1/n)^{3}$.  Then, by applying the  union bound  for all balls and all clients and considering the probability of the conditioning event, we get that \saer($c,d$) completes in $3 \log n$ rounds, with probability at least $1-O(1/n^2)$ .

The next subsection is devoted to the proof of Lemma \ref{lem:fraction_burned}.

\subsection{Proof of Lemma \ref{lem:fraction_burned}}
\label{ssec:lemma_proof}
In this subsection, we   assume that the graph $G((\sC,\sS),E)$ is $\Delta$-regular and $\Delta \geq \eta \log^2 n$ for an arbitrary constant $\eta >0$.
We start by defining the random variables that describe the \saer\ process. 

\begin{definition}
For each round $t \geq 1$ and for each $v \in \sC$,  let $r_t(N(v))$ be the overall number of balls that all the servers in the neighborhood   $N(v) \subseteq \sS$ receives at round $t$; moreover,  let $r_t$ be the maximum number of balls that any server neighborhood receives at round $t$. Formally, 
\begin{equation}
     r_t(N(v)) = \sum_{u \in N(v)}r_t(u) \quad \text{and}\quad r_t=\max_{v \in \sC} r_t(N(v)) \,  .
    \label{def:r_t(N(v))}
\end{equation}
   
\end{definition}
Observe that if a server is burned at a given round then it must have received more than $cd$ balls since  the start of the process. So,  for each $v \in \sC$ it holds that
\begin{equation} \label{eq:Stdet}
    S_t(v) \leq \frac{1}{cd\Delta}\sum_{i=1}^{t}r_i(N(v)) \, .
\end{equation}
   We also name the  expression in the r.h.s. of the inequality above since it will be often used in our analysis.
  
\begin{definition}
\label{def:K_t}
Let
\[K_t(v) = \frac{1}{cd \Delta}\sum_{i=1}^t r_i(N(v)) \
\mbox{ \emph{ and } } \ 
K_t = \frac{1}{cd\Delta}\sum_{i=1}^t r_i \,.
\]
\end{definition}
Notice that   the above definitions and  \eqref{eq:Stdet} easily imply that

\begin{equation}
    S_t \leq K_t
 \ \mbox{ and  } \  K_{t}=K_{t-1}+\frac{1}{cd \Delta}r_t \, , \ \mbox{for each $t \geq 1$} \, .
 \label{eq:bound_Kt_by_Kt-1}
\end{equation}
We next write  the  random variable $r_t(N(v))$ in terms of more ``elementary'' random variables.

\begin{definition}
\label{def:a_t}
For each client $v \in \sC$,  let $a_t^{(i)}(v)$ be the binary random variable indicating whether the $v$'s $i$-th ball  is still \emph{alive} at round $t$, i.e., it has still not been  accepted
by some server at the beginning of round $t$, i.e., 
$$ a_t^{(i)}(v)=
\begin{cases}
1 \text{ if the $v$'s $i$-th ball is still alive at round $t$} \\
0 \text{ otherwise}
\end{cases}
$$
\end{definition}

\begin{definition}
\label{def:z_t}
For each client, $v \in \sC$ and $u \in N(v) \subseteq \sS$, let $z_t^{(i)}(v,u)$ be the binary random variable indicating whether the (random)  contacted server  for the  $v$'s $i$-th ball
at round  $t$ is $u$, i.e.,
\begin{equation}
z_t^{(i)}(v,u)=
\begin{cases}
1 \text{\small \,\,if the contacted server for the $v$'s $i$-th ball at round $t$ is $u$} \\
0 \text{ otherwise}
\end{cases}
\end{equation}
\end{definition}
\noindent According to the above definitions,   for each client $v \in \sC$, we can write
\begin{equation}
 r_{t}\left(N\left(v\right)\right)=\sum_{u \in N(v)}r_t(u)=\sum_{u\in N\left(v\right)}\sum_{w\in N\left(u\right)}\sum_{i=1}^{d}a_{t}^{\left(i\right)}\left(w\right)\cdot z_{t}^{\left(i\right)}\left(w,u\right).
\label{eq:r_neigh}
\end{equation}

We remark that the variable $z_{t}^{\left(i\right)}\left(w,u\right)$ is defined at every round $t \geq 1$, even when the corresponding request $i$ of node $w$ has been already accepted in some previous round.  The above random variables have the following useful properties.

\begin{lemma}
\label{fact}
\begin{enumerate}
    \item For each    $t \geq 1$, $i \in [d]$, $v \in \sC$ and $u \in \sS$, the random variables $z_{t}^{\left(i\right)}\left(v,u\right)$  and  $a_{t}^{\left(i\right)}\left(v\right)$ are mutually independent.
    \item Let $s_0 = 1$. For each $v \in \sC$  and any choice of positive reals $s_j\leq 1$ for $j=1,\dots,t-1$, it holds
\begin{equation}
\Prc{a_{t}^{\left(i\right)}\left(v\right)=1\,|\,S_1 \leq s_1,\dots,S_{t-1}\leq s_{t-1}}\leq \prod_{j=0}^{t-1}s_{j} \,. \label{eq:bound_on_request_sent}
\end{equation}
\item The random variables $\left\{z_t^{(i)}(v,u)\right\}_{v \in \sC,u \in N(v),i \in [d]}$ are negatively associated \footnote{The definition of negative association is given in Definition \ref{def:neg_association} in   Appendix \ref{app:maths}.  This property allows to apply concentration bounds    (see Theorem \ref{thm:chernoff_neg_cor} in the Appendix).}.
\end{enumerate}
\end{lemma}
\begin{proof}[Proof of Lemma \ref{fact}]
Claim $1$ follows from the observation that \saer\  is non-adaptive and symmetric and,  hence,  at each round, each client
$v \in \sC$ chooses the (random) destination of its  $i$-th request regardless of the value of $a_{t}^{\left(i\right)}\left(v\right)$
while the latter determines whether the request is  really sent or not. \\
As for Claim 2, notice  that   $a_{t}^{\left(i\right)}\left(v\right)=1$ iff $v$'s
$i$-th request have been rejected at each previous 
round, and this happens iff the destination of the $i$-th request is a  burned server. \\
Finally,  Claim 3 follows from  the fact that, for each $v \in \sC$, if $z_t^{(i)}(v,u)=1$ for $u \in N(v)$ then, for any $u' \in \sC$ with $u' \neq u$,  it holds that $z_t^{(i)}(v,u')=0$.  Moreover, for each fixed $u \in \sS$ the random variables $\{z_t^{(i)}(v,u), v \in \sC, i \in [d]	\}$ are independent.
\end{proof}

%%%%%%%

\smallskip
\noindent
\textbf{Step-By-Step Analysis via Induction.}
We first  consider the first round of the process and give the following bound on the maximum number of balls a client neighborhood can receive.

\begin{lemma}[First round] \label{lem:firstround} 
For all $c,d\geq 1$, w.h.p. 
\begin{equation} \label{eq:upper_bound_r_1} 
    r_1 \leq 2d \Delta \, 
     \ \mbox{ \emph{ and } } \ 
K_1 \leq \frac{2}{c}\, .
\end{equation}
\end{lemma}
\begin{proof}[Proof of Lemma \ref{lem:firstround}] For each $v \in \sC$ we can write $r_1(N(v))$ as in \eqref{eq:r_neigh} and since  each $z_1^{(i)}(w,u)$ is a Bernoulli random variable of parameter $1/\Delta$,  $\Expcc{r_1(N(v))}=d \Delta$.
Thanks to Claim 3  of Lemma \ref{fact}, we can apply Chernoff bound for negatively associated random variables with $\varepsilon=1$ (Theorem \ref{thm:chernoff_neg_cor} in the Appendix) and get
\begin{equation}
    \Prc{r_1(N(v))\geq 2d \Delta } \leq e^{-\frac{d\Delta}{3}}.
    \label{eq:bound_r_1_little}
\end{equation}
According to Definition \ref{def:r_t(N(v))} and Definition \ref{def:K_t}, from \eqref{eq:bound_r_1_little} and by the 
union bound, we get that 
\begin{align}
\label{eq:bound_r1}
    \Prc {r_1 \leq 2 d\Delta }\geq 1-ne^{-\frac{d \Delta}{3}} \  \text{ and } \  \Prc{K_1 \leq \frac{2}{c}} \geq 1-ne^{-\frac{d \Delta}{3}} \,.
    \end{align}
Since $d \geq 1$ and $\Delta\geq \eta \log^2 n$, the above bounds conclude the proof.
\end{proof}

The next   result  is a key step of the proof of Lemma \ref{lem:fraction_burned}.  We look at a fixed  round $t \geq 2$ of the random process and    derive, for each client $v \in \sC$, an upper bound in concentration  on the random variable  $r_t(N(v))$, assuming some fixed  bounds on the variables $K_1,\dots,K_{t-1}$.
This bound shows that, conditional on the bound sequence above, 
the number of alive balls in $N(v)$ decreases, at each round $t$,  by a factor that explicit depends on the fraction of burned servers at round $t-1$.

\begin{lemma}[Round $t \geq 2$ by induction] \label{claim:chernoff}
Let  $v \in \sC$ and $k_0=1$. For  each choice of positive reals $k_j \leq 1$ with $j=1,\dots,t-1$ and for all $c,d \geq 1$,
\begin{equation}
\label{eq:bound_exp}
    \Expcc{r_t(N(v)) \mid K_1 \leq k_1,\dots K_{t-1}\leq k_{t-1}} \leq \Delta d \prod_{j=0}^{t-1}k_j \, .
\end{equation}
Moreover, for any $\mu>0$ such that $\mu \geq d\Delta \prod_{j=0}^{t-1}k_j$,
\begin{equation} \label{eq:chernoff_claim}
    \Prc{r_t(N(v)) \geq 2\mu  \mid K_1 \leq k_1, \dots, K_{t-1} \leq k_{t-1} } \leq e^{-\frac{\mu}{3}} \, .
    \end{equation}

\label{lem:general_round_t}
    \end{lemma}
\begin{proof}[Proof of Lemma \ref{lem:general_round_t}]
By expressing   $r_t(N(v))$ as the sum    in \eqref{eq:r_neigh}, we can   apply  the first two claims  in Lemma \ref{fact} and get 
\begin{equation}
    \Expcc{r_t(N(v)) \mid K_1 \leq k_1,\dots,K_{t-1}\leq k_{t-1}} \leq d\Delta \prod_{j=0}^{t-1}k_j \, .
\end{equation}
In order to get the claimed bound in concentration, 
we need to  apply the Chernoff bound to the sum of random variables of the form $a_t^{(i)}(w) \cdot z_t^{(i)}(w,u)$.
To this aim, we know that for each $u \in \sS$ and each $w \in N(u)$, $z_t^{(i)}(w,u)$ is a Bernoulli random variables of parameter $1/\Delta$. However,  the distributions of $a_t^{(i)}(w) \cdot z_t^{(i)}(w,u)$ are rather difficult to analyze since there are several correlations among the random variables  in $\{ a_t^{(i)}(w) \, : \ w \in \sC,  \, i \in [d]\}$.
To cope with this issue, we exploit Claim 2 of Lemma \ref{fact}  and  construct $nd$ ad-hoc independent Bernoulli random variables, $\left(X_t^{(i)}(w) \right)_{w \in \sC,i \in [d]}$ for which: 
\begin{equation}
    \Prc{X_t^{(i)}(w)=1 \mid  K_1 \leq k_1, \dots, K_{t-1}\leq k_{t-1}}=\prod_{j=0}^{t-1}k_j \, 
    \label{def:X_t}
\end{equation}
and such that each $X_t^{(i)}(w)$ stochastically dominates   $a_t^{(i)}(w)$.
Formally, thanks to  \eqref{def:X_t} and Claim $2$ of Lemma \ref{fact}, we can define a \emph{coupling}\footnote{See for instance Chapter $05$ of \cite{LPW06}.} between $a_t^{(i)}(w)$ and $X_t^{(i)}(w)$    
such that 
\begin{align}
    \mathbf{Pr}\bigr( \bigcap_{i \in [d], w \in \sC} \left\{ a_t^{(i)}(w) \leq X_t^{(i)}(w) \right\}  \mid K_1\leq k_1,\dots,K_{t-1}\leq k_{t-1} \bigr)=1.\,\,
    \label{eq:prop_coupling}
\end{align}
The detailed construction  of the above coupling is given in   Appendix \ref{appendix:claim_chernoff}.
By using the  coupling, from \eqref{eq:prop_coupling}, we get  
\begin{align}
& \Prc {r_{t}\left(N\left(v\right)\right)\geq 2 \mu \mid K_1 \leq k_1,\dots, K_{t-1} \leq k_{t-1}}  \notag \\
& \leq \mathbf{Pr} \biggr( \sum_{i=1}^{d}\sum_{u\in N\left(v\right)}\sum_{w\in N\left(u\right)}X_t^{\left(i\right)}\left(w\right)\cdot z_{t}^{\left(i\right)}\left(w,u\right)\geq 2\mu  \mid K_1 \leq k_1, \dots,K_{t-1}\leq k_{t-1} \biggr)
\leq  \ e^{-\frac{\mu}{3}} \, ,\label{eq:main_chernoff_bound}
\end{align}
where $\mu>0$ is any positive real that satisfies $\mu \geq d \Delta \prod_{j=0}^{t-1}k_j$.
In detail, \eqref{eq:main_chernoff_bound} follows from  \eqref{eq:prop_coupling} and   the inequality \eqref{eq:main_chernoff_bound} follows by applying the Chernoff bound with $\varepsilon=1$ for negatively associated random variables (see Theorem \ref{thm:chernoff_neg_cor} in the Appendix). Indeed, Claim $3$ of Lemma \ref{fact} and \eqref{def:X_t} imply that the random variables \[\left(X_t^{(i)}(w)\cdot z_t^{(i)}(w,u)\right)_{i \in [d],u\in N(v),w \in N(u)},\] conditioning on the event $\{K_1,\leq k_1,\dots,K_{t-1}\leq k_{t-1}\}$, are distributed as  Bernoulli's one  of parameter $\prod_{j=0}^{t-1}k_j/\Delta$ and they  are negatively associated (see Definition \ref{def:neg_association} in the Appendix). 

\end{proof}

\smallskip
\noindent
\textbf{Wrapping up: Process Analysis in Two Time Stages.}
Lemmas \ref{lem:firstround} and \ref{lem:general_round_t} provide the decreasing rate  of the number of   $r_t(N(v))$ for each $v \in \sC$
conditioning on the events ``$K_j \leq k_j$'' for  a \emph{generic} sequence $k_j$ ($j=1, \ldots, t-1$).

We now need to derive the specific  sequence of  $k_j$   that effectively works
for our process  and that leads to   Lemma  \ref{lem:fraction_burned}. 
Moreover, we notice  that \eqref{eq:bound_exp} in Lemma \ref{lem:general_round_t} (only) allows
a sufficiently strong concentration as long as the bound $\mu$ we can use on the expectation of  $r_t(N(v))$ keeps of order $\Omega(\log n)$, while we clearly need to get an effective  concentration bound until   this value reaches    $0$.

To address    the issues above, we split our analysis in two time stages.
Roughly speaking,   the first stage  proceeds as long as the expectation  of  $r_t(N(v))$ is  $\Omega(\log n)$ and we show it is characterized by an exponential decreasing of $r_t(N(v))$ (see Lemma \ref{claim:recurrence} and Lemma \ref{lem:caseI}).
In the second stage, our technical goal is instead  to show that the fraction of burned nodes in $N(v)$ keeps bounded by some constant $<1$, while neglecting the decreasing rate of  the balls received by $N(v)$ (since we cannot anymore get strong concentration bounds on this random variable). Essentially,  our analysis shows that: i) the process starts this second stage when the expectation   of $r_t(N(v))$ is $\Theta(\log n)$; ii)  during a subsequent window of $\bigO(\log n)$   rounds, the fraction of burned nodes in $N(v)$ keeps bounded by some constant $<1$ and, hence,  all the alive requests will be successfully assigned, w.h.p.

As for the first stage, we consider the   sequence
  $\{\gamma_t\}_{t \in \mathbb{N}}$ defined by the following recurrence
\begin{equation}
 \begin{cases}
 \gamma_0=1 \,  \\
\gamma_t=\frac{2}{c} \sum_{i=1}^t\prod_{j=0}^{i-1}\gamma_j \ \mbox{ for } t \geq 1 \, .
\end{cases}
\label{def:gamma}
\end{equation}
In Appendix \ref{sec:proof_clam_recurrence},  we will prove the following properties.

\begin{lemma} \label{claim:recurrence}
For each $c>1$, let $\{\gamma_t\}_{t \geq 0}$ be the sequence defined by the recurrence \eqref{def:gamma}. Then, if we take $\alpha \geq 2$ such that $\frac{2}{c}\leq \frac{1}{\alpha^2}$, we have the following facts:
\begin{itemize}
    \item $\{\gamma_t\}_{t \in \mathbb{N}}$ is increasing;
    \item for each $t \geq 1$, $\gamma_t \leq \frac{1}{\alpha}$;
    \item for each $t \geq 1$, $\prod_{j=0}^{t-1}\gamma_j \leq \frac{1}{\alpha^t}$. 
\end{itemize}
\end{lemma}
The next  lemma   provides some useful  concentration bounds on the random variables $K_t$ and 
$r_t(N(v))$ for the first stage.

\begin{lemma}[Stage I: Fast decreasing of $r_t(N(v))$]
\label{lem:caseI} For any $c \geq 32$ and for a sufficiently large $n$, an integer $T=\bigO \left(\log(d\Delta/\log n)\right)$ exists such that,
 for each $0 \leq t < T$, 
       \begin{align}
       \label{eq:epxdecrrt}
    & \mathbf{Pr}\biggr(r_t \leq 2 d\Delta \prod_{j=0}^{t-1}\gamma_j \mid K_1 \leq \gamma_1,\dots,K_{t-1}\leq \gamma_{t-1}\biggr) \geq 1-\frac{1}{n^3} \\
    \label{eq:PhaseI} &\text{and} \quad
        \Prc{K_t \leq \gamma_t \mid K_1 \leq \gamma_1,\dots,K_{t-1}\leq \gamma_{t-1}} \geq 1-\frac{1}{n^3} \, , \quad
      \end{align}
     where $\gamma_t$  
is  defined  by the recurrence \eqref{def:gamma}.
\end{lemma}

\begin{proof}[Proof of Lemma \ref{lem:caseI}]
We consider   $\gamma_t$  as in \eqref{def:gamma} and apply  Lemma \ref{claim:chernoff} with  $\mu =\Delta d \prod_{j=0}^{t-1}\gamma_t$. We get, for each $v \in \sC$, 
\begin{align*}
    \Prc {r_t(N(v))\geq 2 d\Delta \prod_{j=0}^{t-1} \gamma_j \mid K_1 \leq \gamma_1,\dots,K_{t-1} \leq \gamma_{t-1}} \leq e^{-\frac{1}{3}d \Delta \prod_{j=0}^{t-1}\gamma_j} \, .
\end{align*}
  From     \eqref{eq:bound_Kt_by_Kt-1}, we know that
  $K_{t}=K_{t-1}+ \frac{1}{cd \Delta}r_t$, so,  using 
  the  union bound  over  all clients $v$, we get 
\begin{align}
    &\Prc{K_t \leq \gamma_t \mid K_1 \leq \gamma_1, \dots, K_{t-1} \leq \gamma_{t-1}} \geq \notag \notag 
     \Prc {r_t \leq 2\Delta d \prod_{j=0}^{t-1}\gamma_j \mid K_1 \leq \gamma_1,\dots,K_{t-1}\leq \gamma_{t-1}} \notag \\ &\geq 1-ne^{-\frac{1}{3}\Delta d \prod_{j=0}^{t-1}\gamma_j},
    \label{eq:bound_K_t}
\end{align}
where in the first inequality we also used      the definition of $\gamma_t$ given  in  \eqref{def:gamma}.
Lemma \ref{claim:recurrence} and the fact that $\Delta \geq \eta \log^2 n$ ensure   that 
 for a sufficiently large $n$ we can take $T\geq 1$ as the smallest positive for which
\begin{equation}
    \Delta d \prod_{j=0}^{T-1}\gamma_j  \leq 12 \log n\,
    \label{eq:condition_T}
\end{equation}
thus 
\begin{equation}
\Delta d \prod_{j=0}^{t-1}\gamma_j > 12 \log n \quad \text{for all $t < T$.}
\label{eq:condition_T_sostitute}
\end{equation}
Moreover, again from Lemma \ref{claim:recurrence}, if we take $c \geq 32$ we have that
$\prod_{j=0}^{T-1}\gamma_j \leq \left(1/4\right)^{T}$
and so, from \eqref{eq:condition_T},  we can say that such a $T$ verifies
\begin{equation*}
 T \leq \frac{1}{2} \log \frac{d \Delta}{12 \log n} \, .
\end{equation*}
Finally, using  \eqref{eq:condition_T_sostitute}   in  \eqref{eq:bound_K_t}, we get \eqref{eq:PhaseI}
for 
 each $t < T$.
\end{proof}

%%%%%%%%%%%%%%%%%%%%%%%
The next result   characterizes the number  of  burned servers along the second, final stage of our process analysis.  

\begin{lemma}[Stage II: The fraction of burned servers keeps small]
\label{lem:caseII}  For any $c \geq \max(32, 288/(\eta d))$ and for a sufficiently large $n$, an integer  $T \geq 1$ exists (it can be the same stated in the previous lemma) such that, for  each  $t$  in the range $[T,3 \log n]$,
    \begin{align}
        \mathbf{Pr}\bigr(K_t \leq \delta_t \mid K_1 \leq \gamma_1,\dots,K_{T-1}\leq \gamma_{T-1},K_T \leq \delta_T,\dots,K_{t-1}\leq \delta_{t-1} \bigr) \geq 1- \frac{1}{n^3} \, , 
        \label{eq:main_high_prob_2}
    \end{align}
 where $\gamma_t$ is defined in \eqref{def:gamma} and $\delta_t$
is defined  by the recurrence  
\begin{equation} \label{def:delta}
\delta_t=\frac{1}{4}+\frac{24t \log n}{cd \Delta} \, ,  \text{ for $t \geq T$.}
\end{equation}
\end{lemma}

%%%%%. PROOF OF THE SECOND PHASE
\begin{proof}[Proof of Lemma \ref{lem:caseII}]
  As in the proof of Lemma \ref{lem:caseI}, let  $T$ be the first integer such that
\begin{equation} \Delta d \prod_{j=0}^{T-1}\gamma_j \leq 12 \log n \,
.
\label{eq:condition_T_2}
\end{equation}
  Observe first  that, for each $t \leq 3 \log n$, since $\Delta \geq \eta \log^2 n$, for  $c \geq 288/(d \eta)$, we have that $\delta_t \leq 1/2$.
So, for each $t$ such that $T \leq t \leq 3 \log n$, \eqref{eq:condition_T_2} and  Lemma \ref{claim:chernoff} imply 
\begin{align}
    \Expcc{r_t(N(v)) \mid K_1 \leq \gamma_1,\dots, K_T \leq \delta_T,\dots, K_{t-1}\leq \delta_{t-1}}\leq d\Delta \prod_{j=0}^{T-1}\gamma_t \prod_{i=T}^{t-1}\delta_i \leq d \Delta \prod_{j=0}^{T-1}\gamma_t \leq 12 \log n. \notag
\end{align}
Hence,  we can   apply
\eqref{eq:chernoff_claim} in Lemma  \ref{claim:chernoff}  with \\ $(k_1,\dots,k_{T-1})=(\gamma_1,\dots,\gamma_{T-1})$ and $(k_{T},\dots,k_{t-1})=(\delta_T,\dots,\delta_{t-1})$ and $\mu=12 \log n$, obtaining 
\begin{align}
    \mathbf{Pr} \bigr(r_t(N(v)) \geq  24 \log n \mid K_1 \leq \gamma_1,\dots K_{T-1} \leq \gamma_{T-1},K_T \leq \delta_T,\dots,K_t \leq \delta_t \bigr)\leq \frac{1}{n^4}. \notag
\end{align}
 Finally, from \eqref{eq:bound_Kt_by_Kt-1} we know that $K_t = K_{t-1}+\frac{1}{cd \Delta}r_t$, so using the definition of $\delta_t$ in \eqref{def:delta} and the union bound over all the clients $v$, we get  \eqref{eq:main_high_prob_2} for each $t \geq T$.
\end{proof}

  Lemma \ref{lem:caseI} and \ref{lem:caseII}   imply  Lemma \ref{lem:fraction_burned}. Indeed, for the chain rule, taking $T'=\lfloor 3 \log n \rfloor $, and $c \geq \max(32,288/(\eta d ))$, we get 
\begin{align} 
    &\Prc { \cap_{t=1}^{T-1}\{K_t \leq \gamma_t\} \, \bigcap \,  \cap_{t=T}^{T'}\{K_t \leq \delta_t\}} \geq \left(1-\frac{1}{n^3}\right)^{T'} \geq 1-T' \frac{1}{n^3}  \geq 1-\frac{1}{n^{2}}\, ,\label{high_probability}
    \end{align}
where in the first inequality of \eqref{high_probability} we used the chain rule, Lemma \ref{lem:caseI} and \ref{lem:caseII} while the second last inequality of
  \eqref{high_probability}
follows from the binomial inequality, i.e.,  for each $x \geq -1$ and for each $m \in \mathbb{N}$,
$(1+x)^m \geq 1+mx$. 
 
In conclusion, we have shown that $K_t \leq \gamma_t$ for all $t \leq T$ and that
$K_t \leq \delta_t$ for all $t$ such that  $T \leq t \leq 3 \log n$, with probability  at least $1-1/n^{2}$. So, recalling that $S_t \leq K_t$, since $c \geq \max(32,288/(\eta d))$ we have that, from \eqref{def:delta} and Lemma \ref{claim:recurrence}, for all $t$ such that $t \leq 3 \log n$, $S_t \leq \frac{1}{2}$ with probability  at least $1-1/n^{2}$.

%%%%%%%% WORK COMPLEXITY %%%%%%%%%%%

\subsection{The Work Complexity of \saer } \label{ssec:work}
To analyze the overall work performed by \saer\ we proceed  using an 
approach  similar to that in the analysis of the Becchetti et al's algorithm \raes .
For each  $v \in \sC$ and each ball  $i \in [d]$,  recall the random variable  $a_t^{(i)}(v)$   introduced in  Definition \ref{def:a_t}.
Then, the  random variable  counting  the total number of   requests
performed by the clients (plus the relative answers by the servers) 
to assign the $nd$ balls can be easily bounded by 
\begin{equation}
    W \, = \, 2 \cdot \sum_{t=1}^\infty\sum_{i=1}^{d}\sum_{v \in \sC} a_t^{(i)}(v)\,.
\end{equation}
To prove that $W=\mathcal{O}(dn)$ w.h.p., we show that, for any fixed $t \leq 3\log{n}$ and any $k \geq nd/\log n$, it  holds
\begin{equation}
\label{eq:decrease_T}
    \Prc{\sum_{i=1}^d \sum_{v \in \sC}a_t^{(i)}(v) > \frac{4}{5}k \mid \sum_{i=1}^d \sum_{v \in \sC}a_{t-1}^{(i)}(v)=k} \leq e^{-\frac{k}{25cd}}\,.
\end{equation}
To this aim, we use the  \emph{method of bounded differences} (see Theorem \ref{thm:bounded_differences} in the Appendix).
We notice that the random variable $\sum_{i=1}^d \sum_{v \in \sC}a_t^{(i)}(v)$, conditioning on a number $k$ of alive balls at the end of round $t-1$, can be written as $2cd$-Lipschitz function of $k$ independent random variables. Indeed, we define the random variables $w^{(t-1)}$ as the set of alive balls at the end of round $t$ and the random variables $\{Y_i\}_{i \in w^{(t-1)}}$, taking values in $\sS = [n]$, indicating the server-destination in $\sS$ the alive ball tries to connect to at round $t$. The random variables $Y_i$ with $i \in w^{(t-1)}$ are mutually independent, and we can write, given the number $k$ of alive balls at round $t-1$,
\[\sum_{i=1}^d \sum_{v \in \sC}a_t^{(i)}(v)=f(Y_{i_1},\dots,Y_{i_k})\,.\]
The function $f$ is $2cd$-Lipschitz because, if we change one of the values $Y_i$, we are changing the destination of a ball from some   $u_1 \in \sS$ to some    $u_2 \in \sS$. If $u_2$ has received less than $cd$ requests since the start of the process, the change of the destination of the $i$-th ball from $u_1$ to $u_2$ would not have any impact. On the other hand, in the worst case, at most $cd$ balls that try to settle in $u_2$ switch from settled to not settled. A symmetric argument holds for $u_1$ and so if 
\[\mathbf{Y}=(v_{i_1},\dots,v_{i_j},\dots,v_{i_k}) \text{  and } \mathbf{Y'}=(v_{i_1},\dots,v_{i_j}',\dots,v_{i_k})\]
then
\[|f(\mathbf{Y})-f(\mathbf{Y'})| \leq 2cd\,.\]
Lemma \ref{lem:fraction_burned} implies that at each round $t \leq 3 \log n$ the fraction of burned nodes in any node's neighborhood remains bounded by $1/2$ with probability at least $1-1/n^2$. Therefore, for each $t \leq 3\log{n}$ holds
\[\Expcc{\sum_{i=1}^d \sum_{v \in \sC}a_t^{(i)}(v) \mid \sum_{i=1}^d \sum_{v \in \sC}a_{t-1}^{(i)}(v)=k } \leq \frac{k}{2}+\frac{1}{n^2}\]
and we can apply Theorem \ref{thm:bounded_differences} with $\mu=3k/5 $ (since $k \geq nd/\log n$) and $M=k/5$, obtaining \eqref{eq:decrease_T}.

From \eqref{eq:decrease_T} and the chain rule, it  follows that for $T=\Omega\left(\frac{\log \log n}{\log(5/4)}\right)$ rounds the number of alive balls decreases at each round by a factor $4/5$, w.h.p. Hence,  at the end of the  $T$-th round,  the number of alive  balls is smaller than $O(nd/\log n)$, w.h.p. From Theorem \ref{thm:SAER-terminates}, we know that the remaining $nd/\log n$ alive balls are assigned within $\mathcal{O}(\log n)$ round: this implies an additional    work of  $\mathcal{O}(nd)$. Observe that the work  until round $T$ is $nd\sum_{t=1}^T(4/5)^t=\mathcal{O}(nd)$.   Hence, for any constant $d >0$, we get the claimed  linear bound   for the work complexity of \saer($c,d$).

\section{Conclusions and Future Work}
 We devise a simple parallel load-balancing   protocol and we give 
a probabilistic analysis of its performances. The main novelty of this paper lies in considering client-server bipartite graphs that are much more sparse than those considered in previous work. This new setting can  model important network scenarios where proximity and/or trust issues force 
very restricted  sets of  admissible client-server  assignments. From a technical point of view, such  sparse topologies yield new probabilistic issues that make our analysis more challenging than the dense case and rather different from the  previous ones.

Several interesting open questions are left open by our paper. In particular, 
we are particularly intrigued  by  the analysis of our protocol (or   simple variants of it) over  graphs with $o(\log^2n)$ degree and/or  
   in the presence of a dynamic framework where, for instance,   the client requests arrive on line and some random topology change may happen during the protocol execution. As for the latter, we believe that the simple structure of \saer\ can well manage such a dynamic scenario and achieves a metastable regime with good performances.

\clearpage

\bibliographystyle{plain}
\bibliography{BB-biblio}

\begin{thebibliography}{10}

\bibitem{AAFPW97}
James Aspnes, Yossi Azar, Amos Fiat, Serge Plotkin, and Orli Waarts.
\newblock On-line routing of virtual circuits with applications to load
  balancing and machine scheduling.
\newblock {\em J. ACM}, 44(3):486–504, May 1997.

\bibitem{AHKL05}
Baruch Awerbuch, Mohammad~T. Hajiaghayi, Robert~D. Kleinberg, and Tom Leighton.
\newblock Online client-server load balancing without global information.
\newblock In {\em Proceedings of the Sixteenth Annual ACM-SIAM Symposium on
  Discrete Algorithms}, SODA ’05, page 197–206, USA, 2005. Society for
  Industrial and Applied Mathematics.

\bibitem{ABKU94}
Yossi Azar, Andrei~Z. Broder, Anna~R. Karlin, and Eli Upfal.
\newblock Balanced allocations (extended abstract).
\newblock In {\em Proceedings of the Twenty-Sixth Annual ACM Symposium on
  Theory of Computing}, STOC ’94, page 593–602, New York, NY, USA, 1994.
  Association for Computing Machinery.

\bibitem{BCNPT18}
Luca Becchetti, Andrea Clementi, Emanuele Natale, Francesco Pasquale, and Luca
  Trevisan.
\newblock Finding a bounded-degree expander inside a dense one.
\newblock In {\em Proceedings of the Thirty-First Annual ACM-SIAM Symposium on
  Discrete Algorithms}, SODA ’20, page 1320–1336, USA, 2020. Society for
  Industrial and Applied Mathematics.

\bibitem{BBFN10}
P.~{Berenbrink}, A.~{Brinkmann}, T.~{Friedetzky}, and L.~{Nagel}.
\newblock Balls into non-uniform bins.
\newblock In {\em 2010 IEEE International Symposium on Parallel Distributed
  Processing (IPDPS)}, pages 1--10, April 2010.

\bibitem{BBFN12}
Petra Berenbrink, Andr\'{e} Brinkmann, Tom Friedetzky, and Lars Nagel.
\newblock Balls into bins with related random choices.
\newblock {\em J. Parallel Distrib. Comput.}, 72(2):246–253, February 2012.

\bibitem{BCSV06}
Petra. Berenbrink, Artur. Czumaj, Angelika. Steger, and Berthold. Vöcking.
\newblock Balanced allocations: The heavily loaded case.
\newblock {\em SIAM Journal on Computing}, 35(6):1350--1385, 2006.

\bibitem{BFGGHM06}
Petra Berenbrink, Tom Friedetzky, Leslie~Ann Goldberg, Paul Goldberg, Zengjian
  Hu, and Russell Martin.
\newblock Distributed selfish load balancing.
\newblock In {\em Proceedings of the Seventeenth Annual ACM-SIAM Symposium on
  Discrete Algorithm}, SODA ’06, page 354–363, USA, 2006. Society for
  Industrial and Applied Mathematics.

\bibitem{BFHM08}
Petra Berenbrink, Tom Friedetzky, Zengjian Hu, and Russell Martin.
\newblock On weighted balls-into-bins games.
\newblock {\em Theoretical Computer Science}, 409(3):511 -- 520, 2008.

\bibitem{BFLS}
Petra Berenbrink, Tom Friedetzy, Christiane Lammersen, and Thomas Sauwervald.
\newblock Parallel randomized load balancing.
\newblock {\em Unpublished Manuscript}, 2018.

\bibitem{BKSS13}
Petra Berenbrink, Kamyar Khodamoradi, Thomas Sauerwald, and Alexandre Stauffer.
\newblock Balls-into-bins with nearly optimal load distribution.
\newblock In {\em Proceedings of the 25th Annual ACM Symposium on Parallelism
  in Algorithms and Architectures (SPAA 2013)}, pages 326--335, New York, NY,
  USA, 2013. ACM.

\bibitem{BMFS97}
Petra Berenbrink, Friedhelm Meyer auf~der Heide, and Klaus Schr\"{o}der.
\newblock Allocating weighted jobs in parallel.
\newblock In {\em Proceedings of the Ninth Annual ACM Symposium on Parallel
  Algorithms and Architectures}, SPAA ’97, page 302–310, New York, NY, USA,
  1997. Association for Computing Machinery.

\bibitem{BLSZ14}
B.~{Bosek}, D.~{Leniowski}, P.~{Sankowski}, and A.~{Zych}.
\newblock Online bipartite matching in offline time.
\newblock In {\em 2014 IEEE 55th Annual Symposium on Foundations of Computer
  Science}, pages 384--393, Oct 2014.

\bibitem{BCM04}
John~W. Byers, Jeffrey Considine, and Michael Mitzenmacher.
\newblock Geometric generalizations of the power of two choices.
\newblock In {\em Proceedings of the Sixteenth Annual ACM Symposium on
  Parallelism in Algorithms and Architectures}, SPAA ’04, page 54–63, New
  York, NY, USA, 2004. Association for Computing Machinery.

\bibitem{DP09}
Devdatt~P. Dubhashi and Alessandro Panconesi.
\newblock {\em Concentration of measure for the analysis of randomized
  algorithms}.
\newblock Cambridge University Press, 2009.

\bibitem{GP13}
Stephen~D. Gantz and Daniel~R. Philpott.
\newblock Chapter 15 - contingency planning.
\newblock In Stephen~D. Gantz and Daniel~R. Philpott, editors, {\em FISMA and
  the Risk Management Framework}, pages 403 -- 443. Syngress, 2013.

\bibitem{G08}
P.~Brighten Godfrey.
\newblock Balls and bins with structure: Balanced allocations on hypergraphs.
\newblock In {\em Proceedings of the Nineteenth Annual ACM-SIAM Symposium on
  Discrete Algorithms}, SODA ’08, page 511–517, USA, 2008. Society for
  Industrial and Applied Mathematics.

\bibitem{KLMF92}
Richard~M. Karp, Michael Luby, and Friedhelm Meyer auf~der Heide.
\newblock Efficient pram simulation on a distributed memory machine.
\newblock In {\em Proceedings of the Twenty-Fourth Annual ACM Symposium on
  Theory of Computing}, STOC ’92, page 318–326, New York, NY, USA, 1992.
  Association for Computing Machinery.

\bibitem{KP06}
Krishnaram Kenthapadi and Rina Panigrahy.
\newblock Balanced allocation on graphs.
\newblock In {\em Proceedings of the Seventeenth Annual ACM-SIAM Symposium on
  Discrete Algorithm}, SODA ’06, page 434–443, USA, 2006. Society for
  Industrial and Applied Mathematics.

\bibitem{KT09}
Robert Kleinberg, Georgios Piliouras, and \'{E}va Tardos.
\newblock Load balancing without regret in the bulletin board model.
\newblock In {\em Proceedings of the 28th ACM Symposium on Principles of
  Distributed Computing}, PODC ’09, page 56–62, New York, NY, USA, 2009.
  Association for Computing Machinery.

\bibitem{KMS03}
Elias Koutsoupias, Marios Mavronicolas, and Paul Spirakis.
\newblock Approximate equilibria and ball fusion.
\newblock {\em Theory of Computing Systems}, 36(6):683--693, Dec 2003.

\bibitem{LPY19}
Christoph Lenzen, Merav Parter, and Eylon Yogev.
\newblock Parallel balanced allocations: The heavily loaded case.
\newblock In {\em The 31st ACM Symposium on Parallelism in Algorithms and
  Architectures}, SPAA ’19, page 313–322, New York, NY, USA, 2019.
  Association for Computing Machinery.

\bibitem{LW11}
Christoph Lenzen and Roger Wattenhofer.
\newblock Tight bounds for parallel randomized load balancing: Extended
  abstract.
\newblock In {\em Proceedings of the 43rd Annual ACM Symposium on Theory of
  Computing (STOC 2011)}, pages 11--20, New York, NY, USA, 2011. ACM.

\bibitem{LPW06}
David~A. Levin, Yuval Peres, and Elizabeth~L. Wilmer.
\newblock {\em {Markov chains and mixing times}}.
\newblock American Mathematical Society, 2006.

\bibitem{ACMR98}
Adler Micah, Chakrabarti Soumen, and Rasmussen~Lars E.
\newblock Parallel randomized load balancing.
\newblock {\em Random Struct. Algorithms}, 13(2):159--188, 1998.

\bibitem{MU17}
Michael Mitzenmacher and Eli Upfal.
\newblock {\em Probability and Computing: Randomization and Probabilistic
  Techniques in Algorithms and Data Analysis}.
\newblock Cambridge University Press, USA, 2nd edition, 2017.

\bibitem{R02}
Harald R\"{a}cke.
\newblock Minimizing congestion in general networks.
\newblock In {\em Proceedings of the 43rd Symposium on Foundations of Computer
  Science}, FOCS ’02, page 43–52, USA, 2002. IEEE Computer Society.

\bibitem{V03}
Berthold V\"{o}cking.
\newblock How asymmetry helps load balancing.
\newblock {\em J. ACM}, 50(4):568–589, July 2003.

\bibitem{W07}
Udi Wieder.
\newblock Balanced allocations with heterogenous bins.
\newblock In {\em Proceedings of the Nineteenth Annual ACM Symposium on
  Parallel Algorithms and Architectures}, SPAA ’07, page 188–193, New York,
  NY, USA, 2007. Association for Computing Machinery.

\bibitem{ZSZF10}
C.~{Zhang}, J.~{Sun}, X.~{Zhu}, and Y.~{Fang}.
\newblock Privacy and security for online social networks: challenges and
  opportunities.
\newblock {\em IEEE Network}, 24(4):13--18, July 2010.

\end{thebibliography}
\appendix

%%%%%%%%%%%%%%%%%% TOOLS %%%%%%%%%%%%%%%%%%%%
\newpage
\section{Mathematical tools} \label{app:maths}

\begin{definition}[Negative association, \cite{DP09}]
\label{def:neg_association}
The random variables $X_i$, $i \in [n]$ are \emph{negatively associated} if for all disjoint subsets $I, J \subseteq [n]$ and all nondecreasing
functions $f$ and $g$,
\begin{equation}
    \Expcc{f(X_i,i \in I)g(X_j,j \in J)} \leq \Expcc{f(X_i, i \in I)}\Expcc{g(X_j, j \in J)}\,.
\end{equation}
\end{definition}

\begin{theorem}[Chernoff for negatively associated random variables, \cite{DP09}]
\label{thm:chernoff_neg_cor}
Let $X_1,\dots,X_n$ a family of random variables in $\{0,1\}$ negatively associated and $X=X_1+ \dots + X_n$. Let $p_i=\Expcc{X_i}$ and define $\mu=\Expcc{X}=p_1+\dots+p_n$. Then, for any reals $\varepsilon \in (0,1]$ 
\[\Prc{X \geq(1+\varepsilon)\mu}\leq e^{-\frac{\varepsilon^2}{3}\mu}\,.\]
\end{theorem}

\begin{theorem}[Method of bounded differences, \cite{DP09}]
\label{thm:bounded_differences}
Let $\mathbf{Y}=(Y_1,\dots,Y_m)$ be independent random variables, with $Y_j$ taking values in the set $A_j$. Suppose the real-valued function $f$ defined on $\prod_j A_j$ satisfies the Lipschitz condition with coefficients $\beta_j$, i.e.
\[|f(\mathbf{y})-f(\mathbf{y'})| \leq \beta_j\]
whenever vectors $\mathbf{y}$ $\mathbf{y'}$ differs only in the $j$-th coordinate. Let $\mu$ an upper bound to the expected value of r.v. $f(\mathbf{Y})$. Then, for any $M>0$, it holds that
\[\Prc{f(\mathbf{Y})-\mu \geq M }\leq e^{-\frac{2M^2}{\sum_{j=1}^{m}\beta_j}}\,.\]
\end{theorem}

\section{Proof of Lemma \ref{claim:recurrence}} \label{sec:proof_clam_recurrence}

From \eqref{def:gamma}, we can state that for each $t \geq 1$
\begin{equation}
\label{eq:gamma_recurrence}
    \gamma_{t+1}=\gamma_{t}+\frac{2}{c}\prod_{j=0}^{t}\gamma_t,
\end{equation}
and so the sequence is increasing.
Now we want to prove, by induction, that each term of the sequence verifies $\gamma_t \leq \frac{1}{\alpha}-\frac{1}{\alpha^{t+1}}$ for $t\geq 1$. From that, clearly follows the Lemma. We notice that the hypothesis holds for $\gamma_1$, since $\alpha^2 \leq c/2$ and $ \alpha \geq 2$. Now, assuming that
\begin{equation}
\gamma_i \leq \frac{1}{\alpha}-\frac{1}{\alpha^{i+1}} \quad \text{for each } i \leq t
\label{eq:hyp_ind_claim}
\end{equation} 
we show that $\gamma_{t+1}\leq \frac{1}{\alpha}-\frac{1}{\alpha^{t+1}}$. From \eqref{eq:gamma_recurrence} and from \eqref{eq:hyp_ind_claim} we get that
\[\gamma_{t+1}-\gamma_{t} = \frac{2}{c}\prod_{j=0}^{t}\gamma_j \leq \frac{2}{c}\frac{1}{\alpha^{t-1}} \leq \frac{1}{\alpha^{t+2}}. \]
Then, we have that
\[\gamma_{t+1} \leq \gamma_t + \frac{1}{\alpha^{t+2}} \leq \frac{1}{\alpha}-\frac{1}{\alpha^{t+1}}+\frac{1}{\alpha^{t+2}} \leq \frac{1}{\alpha}-\frac{1}{\alpha^{t+2}},\]
since $\alpha \geq 2$.

%%%%%%%%%%%% COUPLING %%%%%%%%%%%%%%%

\section{Construction of the coupling in Lemma \ref{claim:chernoff}} \label{appendix:claim_chernoff}

\begin{lemma}
\label{claim:construction_coupling} In the setting of Lemma \ref{claim:chernoff}, for each $t \geq 1$ and for any choice of positive reals $k_j \leq 1$, with $j=1,\dots,t-1$, we can define a coupling $\left(a_t^{(i)}(w),X_t^{(i)}(w)\right)_{i \in [d],w \in \sC}$ such that
\begin{align}
    \mathbf{Pr} \bigr( \bigcap_{i \in [d], w \in \sC} \{ a_t^{(i)}(w) \leq X_t^{(i)}(w) \} \mid K_1\leq k_1,\dots,K_{t-1}\leq k_{t-1} \bigr)=1
    \label{eq:prop_coupling_2}
\end{align}
where $\left(X_t^{(i)}(w) \right)_{i \in [d], w \in \sC}$ are $nd$ independent Bernoulli random variables in $\{0,1\}$ such that
\begin{equation}
    \Prc{X_t^{(i)}(w)=1 \mid  K_1 \leq k_1, \dots, K_{t-1}\leq k_{t-1}}=\prod_{j=0}^{t-1}k_j.
\end{equation}
\end{lemma}

\begin{proof}[Proof of Lemma \ref{claim:construction_coupling}]
To define the coupling, we consider $nd$ uniform and independent random variables in $[0,1]$, $U_w^{(i)}$ with $i \in [d]$ and $w \in \sC$.
Given $t \geq 1$, $w \in \sC$ and $i \in [d]$,  we define the following set of random variables:
\begin{equation}
    A_{t,i,w}=\{a_t^{(j)}(v): \, j<i \text{ for } v=w \text{ and } j \in [d] \text{ for } v<w\}
\end{equation}
which is nothing but   the previous random variables of $a_t^{(i)}(w)$ according to the following sorting ($w=v_h$ for some $h$):
\begin{center}
    \begin{tabular}{cc}
         &$a_t^{(1)}(v_1), a_t^{(2)}(v_1), \dots, a_t^{(d)}(v_1)$  \\
         &$a_t^{(1)}(v_2), a_t^{(2)}(v_2), \dots, a_t^{(d)}(v_2)$ \\
         &$\dots$ \\
          &$a_t^{(1)}(v_n), a_t^{(2)}(v_n), \dots, a_t^{(d)}(v_n).$
    \end{tabular}
    \end{center}
In the next definition, we will improperly use the term $A_{t,i,w}$ to denote the event in which the random variables $a_t^{(j)}(v)$ of subset $A_{t,i,w}$ take any  fixed values in $\{0,1\}$.
For each $i \in [d]$ and $w \in \sC$, given $\{K_1 \leq k_1,\dots,K_{t-1}\leq k_{t-1}\}$ we define the following two events
\begin{align*}
&H_{t,i,w}=\{U_w^{(i)} \leq \prod_{j=0}^{t-1}k_j \}
\\
&K_{t,i,w} \ 
 = \{ U_w^{(i)} \, \leq \, \Prc{a_t^{(i)}(w)=1 \mid A_{t,i,w},K_1 \leq k_1,\dots,K_{t-1}\leq k_{t-1}}\} \, .
\end{align*}
Now we can define the coupling. For $h_{i,w},k_{i,w} \in \{0,1\}$
\begin{align}
&\mathbf{Pr}\biggr( \bigcap_{i \in [d], w \in \sC} \bigr\{ \bigr(X_t^{(i)}(w),a_t^{(i)}(w)\bigr)=(h_{i,w},k_{i,w})\bigr\}\notag    \mid K_1 \leq k_1,\dots,K_{t-1}\leq k_{t-1}  \biggr)= \notag
\\
    &\mathbf{Pr} \biggr( \bigcap_{i \in [d],w \in \sC}\bigr\{\bigr(\mathbb{1}_{H_{t,i,w}},\mathbb{1}_{K_{t,i,w}}\bigr)=(h_{i,w},k_{i,w})\bigr\}   \mid K_1 \leq k_1,\dots,K_{t-1}\leq k_{t-1} \biggr)\,.
    \label{def:coupling}
\end{align}
Now we show that the coupling is well defined, i.e. the marginal laws are the same of $X_t^{(i)}(w)$ and $a_t^{(i)}(w)$. It's trivial that 
\begin{align}
&\mathbf{Pr} \bigr( \bigcap_{i \in [d], w \in \sC} \bigr\{X_t^{(i)}(w)=h_{i,w} \bigr\}\notag \mid K_1 \leq k_1,\dots,K_{t-1}\leq k_{t-1} \bigr)=\notag \\
\notag
    &\mathbf{Pr} \bigr(  \bigcap_{i     \in [d],w \in \sC} \bigr\{ \mathbb{1}_{H_{t,i,w}}=h_{i,w} \bigr\}  \mid K_1 \leq k_1,\dots,K_{t-1}\leq k_{t-1} \bigr) \,.
\end{align}
We have also that
\begin{align}
    &\mathbf{Pr} \bigr( \bigcap_{i \in [d],w \in \sC}\bigr\{\mathbb{1}_{K_{t,i,w}}=k_{i,w}\bigr\}    \mid K_1 \leq k_1,\dots,K_{t-1}\leq k_{t-1}\bigr)=
\\
    &\prod_{i,w:k_{i,w}=1} \Prc{K_{t,i,w} \mid K_1 \leq k_1,\dots,K_{t-1}\leq k_{t-1}} \cdot  \prod_{i,w:k_{i,w}=0} \Prc{K_{t,i,w}^C \mid K_1 \leq k_1,\dots,K_{t-1}\leq k_{t-1}}=  \label{eq:indipendence}  \\ 
    &\prod_{i,w:k_{i,w}=1}\Prc { a_t^{(i)}(w)=1 \mid A_{t,i,w}, K_1 \leq k_1,\dots,K_{t-1}\leq k_{t-1} } \cdot \notag \\ &\prod_{i,w:k_{i,w}=0}\Prc{a_t^{(i)}(w)=0 \mid A_{t,i,w}, K_1 \leq k_1,\dots,K_{t-1}\leq k_{t-1}}= 
    \\ &
    \prod_{i,w} \mathbf{Pr}\bigr(a_t^{(i)}(w)=k_{i,w} \mid A_{t,i,w}, K_1 \leq k_1,\dots,K_{t-1}\leq k_{t-1}\bigr)= \notag \\
     &\textbf{Pr}\bigr( \bigcap_{i \in [d],w \in \sC} \bigr\{a_t^{(i)}(w)=k_{i,w} \bigr\} \mid K_1 \leq k_1,\dots,K_{t-1}\leq k_{t-1}\bigr). \label{eq:chain rule}
\end{align}
 \eqref{eq:indipendence} follows by the independence of the random variables $U_w^{(i)}$ with $i \in [d]$ and $w \in \sC$. \eqref{eq:chain rule} follows by the chain rule with the same sorting adopted in the definition of $A_{t,i,w}$. It's easy to see that the coupling satisfies \eqref{eq:prop_coupling_2}.  Indeed
\begin{align*} \mathbf{Pr} \bigr( \bigcap_{i \in [d],w \in \sC}\bigr\{
a_t^{(i)}(w) \leq X_t^{(i)}(w) \bigr\} \mid K_1 \leq k_1,\dots,K_{t-1}\leq k_{t-1} \bigr)=\\
     \mathbf{Pr}\bigr(\bigcap_{i \in [d],w \in \sC}\bigr\{\mathbb{1}_{K_{t,i,w}}  \leq \mathbb{1}_{H_{t,i,w}}\bigr\} \mid K_1 \leq k_1,\dots,K_{t-1}\leq k_{t-1} \bigr)= 1. \end{align*}
     Indeed, for every $i \in [d]$ and $w \in \sC$,
\begin{equation}
    \left\{ K_{t,i,w} \right\} \subseteq \left\{H_{t,i,w}\right\},
\end{equation}
since for each $w \in \sC$
\begin{equation}
    \label{eq_bound_request_cond}
    \Prc {a_t^{(i)}(w)=1 \mid A_{t,i,w}, K_1 \leq k_1,\dots,K_{t-1}\leq k_{t-1}}\leq \prod_{j=0}^{t-1}k_j \, ,
\end{equation}
and, we  can derive the last inequality from the fact that 
\begin{align}
    \Prc{a_t^{(i)}(w)=1 \mid A_{t,i,w},S_1(w)=s_1(w),\dots,S_{t-1}(w)=s_{t-1}(w)}= \prod_{j=0}^{t-1}s_j(w) \, .
\end{align}
\end{proof}

%%%%%%%%% ALMOST-REGULAR %%%%%%%%%%%%%%%%%%%%

\section{Proof of Theorem \ref{thm:SAER-terminates}: Almost-Regular Graphs} \label{ssec:irregular}

In this section, we prove Theorem \ref{thm:SAER-terminates} for any  bipartite graph $G(V=(\sC,\sS),E)$    that satisfies the conditions: $\Deltamin(\sC) \geq \eta  \log^2 n$ and 
$\Deltamax(\sS)/ \Deltamin(\sC) \leq  \rho$.
We will make use of  the   notation and the definitions introduced  in    Section \ref{subsec:thm_1}.  We will only describe  the aspects that differ from the regular case.

Following the approach we used for   the regular case, 
Theorem \ref{thm:SAER-terminates}  is simple consequence of       the next result.

\begin{lemma}
\label{lem:bound_burned_alm_reg}
 Let $\Deltamin(\sC) \geq \eta \log^2 n$ and  $\Deltamax(\sS)/\Deltamin(\sC) \leq \rho$ for arbitrary constants $\eta>0$ and $\rho \geq 1$. Let $d \geq 1$ be an arbitrary constant in $\mathbb{N}$. Then, for any $c\geq \max(32 \rho,288/(\eta d))$ and for a sufficiently large $n$, with probability at least $1-1/n^2$, it holds that, for every $t \leq 3 \log n$, the fraction of burned nodes in \saer$(c,d)$ satisfies 
\[S_t \leq\frac{1}{2}\,.\]
\end{lemma}

\begin{proof}[Proof of Lemma \ref{lem:bound_burned_alm_reg}]
 The fraction of burned nodes in the neighborhood of each $v \in \sC$ is  
\begin{equation}
    S_{t}\left(v\right) \, = \, \frac{\left|\left\{ u:u\in N\left(v\right)\wedge\left(\text{\ensuremath{u} is burned at time \ensuremath{t}}\right)\right\} \right|}{\Delta_v} \, .
\end{equation}
In the non-regular case, the  random variables $K_t$ and $K_t(v)$  can be defined as follows.   For each $v \in \sC$ and each $t \geq 1$,
\begin{equation}
    \label{def:Kt_noreg}
    K_t(v)=\frac{1}{cd \Delta_v} \sum_{i=1}^{t}r_i(N(v)) \quad \text{and}\quad K_t=\max_{v \in V}K_t(v).
\end{equation}
Observe that,  for each $t \geq 1$ and $v \in \sC$,
\begin{equation}
    S_t \leq K_t  \text{ and } K_t(v) \leq K_{t-1} + \frac{1}{cd\Delta_v}r_t(N(v)) \,.
    \label{def:kt_noreg_prop}
\end{equation}
For each $v \in \sC$ and $u \in N(v)$, consider the random variables $z_t^{(i)}(v,u)$ and $a_t^{(i)}(v)$ as introduced for  the regular case   in Definition \ref{def:a_t} and \ref{def:z_t}. The only difference is that, in this more general  setting, $z_t^{(i)}(v,u)$ are Bernoulli random variables of parameter $1/\Delta_v$.  We then  remark that  Lemma \ref{fact} holds   in this setting as well  and its proof is the same.

\smallskip
\noindent \textbf{Step-By-Step Analysis via Induction.}
As in  Subsection \ref{ssec:lemma_proof}, we start by analyzing the bound on the requests received by a   neighborhood of a fixed client in the first round.

\begin{lemma}[First round] For all $c,d \geq 1$ and for any $v \in \sC$, w.h.p.
\begin{equation}
r_1(N(v))\leq 2d \Delta_v \frac{\Deltamax(\sS)}{\Deltamin(\sC)}    
\label{eq:bound_r1_noreg}
\end{equation}
and
\begin{equation}
    K_1 \leq \frac{2}{c}\frac{\Deltamax(\sS)}{\Deltamin(\sC)}\,.
    \label{eq:bound_K1_noreg}
\end{equation}
\label{lem:firstround_noreg}
\end{lemma}
\begin{proof}
 The random variable   $r_1(N(v))$   can be written as in \eqref{eq:r_neigh}.  Then, since for each $v \in \sC$ and $u \in N(v)$, $z_t^{(i)}(v,u)$ is a Bernoulli random variable of parameter $1/\Delta_w$, we get 
\begin{equation}
    \Expcc{r_1(N(v))}=d \sum_{u \in N(v)} \sum_{w \in N(u)} \frac{1}{\Delta_w} \leq d \Delta_v \frac{\Deltamax(\sS)}{\Deltamin(\sC)}.
\end{equation}
From Lemma \ref{fact}, we can apply the Chernoff bound for negatively associated random variables with $\varepsilon=1$ (Theorem \ref{thm:chernoff_neg_cor}) as in \eqref{eq:bound_r_1_little}, thus obtaining
\begin{align}
    \Prc {r_1(N(v)) \geq 2 d \Delta_v \frac{\Deltamax(\sS)}{\Deltamin(\sC)}}  \leq e^{-\frac{1}{3}d \Delta_v \frac{\Deltamax(\sS)}{\Deltamin(\sC)}} \leq e^{-\frac{1}{3}d \Deltamax(\sS)} \,.
    \label{eq:bound_r1_comp}
\end{align}
Observe that the last inequality implies that \eqref{eq:bound_r1_noreg} holds w.h.p., since, in any bipartite graph, we have $\Deltamax(\sS) \geq \Deltamin(\sC) \geq \eta \log^2 n$. Finally, from \eqref{def:Kt_noreg},
\eqref{eq:bound_r1_comp} and from an union bound,  we get  \eqref{eq:bound_K1_noreg}.
\end{proof}

  For each $v \in \sC$ to $r_t(N(v))$, we give an upper bound on $\Expcc{r_t(N(v))}$
  conditional  to   
  some fixed upper bounds on $K_1,\dots,K_{t-1}$. 
\begin{lemma}[Round $t \geq 2$ by induction]
\label{lem:chernoff_noreg}
Let $v \in \sC$ and $k_0=1$. For each choice of positive reals $k_j$ with $j=1,\dots,t-1$ and for all $c,d \geq 1$
\begin{equation}
\Expcc{r_t(N(v)) \mid K_1 \leq k_1,\dots,K_{t-1}\leq k_{t-1}}\leq d \Delta_v \frac{\Deltamax(\sS)}{\Deltamin(\sC)}\prod_{j=0}^{t-1}k_j \,.
\label{eq:bound_Ert_noreg}
\end{equation}
Moreover, for any  $\mu \geq d \Delta_v \cdot (\Deltamax(\sS) /\Deltamin(\sC)) \cdot  \prod_{j=0}^{t-1}k_j,$
\begin{equation}
    \Prc{r_t(N(v)) \geq 2\mu \mid K_1\leq k_1,\dots, K_{t-1}\leq k_{t-1}} \leq e^{-\frac{\mu}{3}}\,.
\end{equation}
\end{lemma}
\begin{proof}
We can proceed as in  the proof of Lemma \ref{claim:chernoff}. The only difference  is   between   equations \eqref{eq:bound_exp} and \eqref{eq:bound_Ert_noreg}:  in the latter,  for each $v \in \sC$ and $u \in N(v)$, the $z_t^{(i)}(v,u)$ are Bernoulli random variables of parameter $1/\Delta_v$. 
\end{proof}

\smallskip
\noindent
\textbf{Wrapping up: Process Analysis in Two Time Stages.}
Lemmas \ref{lem:firstround_noreg} and \ref{lem:chernoff_noreg} (similarly to  Lemmas \ref{lem:firstround} and \ref{claim:chernoff} for the regular case) provide the decreasing rate of the number of alive balls in any fixed $N(v)$, conditioning on the events $"K_j \leq k_j"$, for a generic sequence $k_j$ ($j=1,\dots,t-1$). Now, for the same reasons explained in Subsection \ref{ssec:lemma_proof}, our analysis is organized  in two time stages. In the first phase there is a strong decreasing of $r_t(N(v))$, while in the second phase, which we show it lasts $\bigO(\log n)$ rounds,  our goal is to prove that the fraction of burned nodes in each neighborhood of $v \in \sC$ keeps bounded by some constant $<1$.

As for   the first stage, we consider the sequence $\{\gamma_t'\}_{t \in \mathbb{N}}$ defined by the following recurrence
\begin{equation}
\left\{ \begin{array}{ll}
\gamma'_0=1 \\
\gamma'_t=\frac{2}{c} \frac{\Deltamax(\sS)}{\Deltamin(\sC)}\sum_{i=1}^t \prod_{j=0}^{i-1}\gamma_j'
\end{array}
\right.
\label{def:gamma1}
\end{equation}
This sequence plays the same role as the sequence $\{\gamma_t\}_{t \in \mathbb{N}}$ in   Subsection \ref{ssec:lemma_proof}. Since, by hypothesis  $\Deltamax(\sS)/\Deltamin(\sC) \leq \rho$, the sequence   $\{\gamma_t'\}_{t \in \mathbb{N}}$ has the same properties of $\{\gamma_t\}_{t \in \mathbb{N}}$ described in Lemma \ref{claim:recurrence}. Indeed,  since  $\Deltamax(\sS)/\Deltamin(\sC) \leq \rho= O(1)$, we can take a constant $\alpha$ such that
\[\frac{2}{c}\frac{\Deltamax(\sS)}{\Deltamin(\sC)} \leq \frac{2\rho}{c}\leq \frac{1}{\alpha^2}.\]
In particular, in the analysis we will take  $c \geq 32 \rho$ in order to have $\gamma_t' \leq 1/4$ (see Lemma \ref{claim:recurrence}).

\begin{lemma}[Stage I: Fast decreasing of the active balls]
\label{lem:caseI_noreg} For any $c \geq 32 \rho$  (with $\rho \geq \Deltamax(\sS)/\Deltamin(\sC)$) and for any  sufficiently large $n$,  an integer $T= \bigO\left(\log(d\Deltamax(\sS)/\log n)\right)$ exists such that, for each $0 \leq t < T$,
\begin{align}
\label{eq:caseI.II_noreg}
  \mathbf{Pr}\bigr(\bigcap_{v \in \sC} \bigr\{r_t(N(v)) \leq 2d \Delta_v \frac{\Deltamax(\sS)}{\Deltamin(\sC)} \prod_{j=0}^{t-1}\gamma_j' \bigr\}  \mid K_1 \leq \gamma_1',\dots,K_{t-1}\leq \gamma_{t-1} \bigr) \geq 1-\frac{1}{n^3}
\end{align}
and
\begin{equation}
\label{eq:caseI.I_noreg}
    \Prc{K_t \leq \gamma_t' \mid K_1 \leq \gamma_1',\dots,K_{t-1}\leq \gamma_{t-1}'}\geq 1-\frac{1}{n^3}\,.
\end{equation}

\end{lemma}
\begin{proof}
From Lemma \ref{lem:chernoff_noreg},  we fix  
\[\mu=d\Delta_v \frac{\Deltamax(\sS)}{\Deltamin(\sC)}\prod_{j=0}^{t-1}\gamma_j'\] and, from the  union bound over all     $v \in \sC$,  we get  
\begin{align}
    &\mathbf{Pr} \bigr( \bigcap_{v \in \sC}\bigr\{ r_t(N(v)) \leq 2 \Delta_v d \frac{\Deltamax(\sS)}{\Deltamin(\sC)} \prod_{j=0}^{t-1}\gamma_j' \bigr\}\mid K_1 \leq \gamma_1',\dots, K_{t-1} \leq \gamma_{t-1}'\bigr) \notag \\
    &\geq  \,1-ne^{-\frac{1}{3}d\Delta_v \frac{\Deltamax(\sS)}{\Deltamin(\sC)}\prod_{j=0}^{t-1}\gamma_j'} \geq 1-ne^{-\frac{1}{3}d \Deltamax(\sS)\prod_{j=0}^{t-1}\gamma_j'} \,,
    \label{eq:high_prob_noreg}
\end{align}
where the last inequality in  \eqref{eq:high_prob_noreg} follows from the definition of $\gamma_t'$ in \eqref{def:gamma1} and from \eqref{def:kt_noreg_prop}, the latter stating that,  for each $v \in \sS$, 
\[ K_t(v) \leq K_{t-1}+(1/cd\Delta_v)r_t(N(v)) \, . \]
Similarly to the regular case, we must verify when \eqref{eq:high_prob_noreg} is a high probability. First of all, we recall that the behaviour of $\gamma_t'$ is the same of $\gamma_t$, since $\Deltamax(\sS)/\Deltamin(\sC) \leq \rho$ where $\rho$ is some constant. So, from Lemma \ref{claim:recurrence}, we can take $T \geq 1$ as the smallest  integer for which
\begin{equation}
    \label{eq:condition_T_noreg}
    d \Deltamax(\sS)\prod_{j=0}^{T-1}\gamma_j'\leq 12 \log n \,,
\end{equation}
and, hence, 
\begin{equation}
    \label{eq:condition_T_noreg_substitute}
    d \Deltamax(\sS)\prod_{j=0}^{t-1}\gamma_j'> 12 \log n \quad \text{for each $t <T$.}
\end{equation}
Moreover, again from Lemma \ref{claim:recurrence}, if we take $c \geq 32 \rho$, then $\prod_{j=0}^{T-1}\gamma_j' \leq (1/4)^T$ and so, from \eqref{eq:condition_T_noreg}, we can say that $T$ verifies
\begin{equation}
    T \leq \frac{1}{2} \log \frac{d \Deltamax(\sS)}{12 \log n} \, .
\end{equation}
Finally, by using \eqref{eq:condition_T_noreg_substitute} in \eqref{eq:high_prob_noreg}, we get \eqref{eq:caseI.I_noreg} and \eqref{eq:caseI.II_noreg} for each $t <T$.
\end{proof}

\begin{lemma}[Stage II: The fraction of burned servers keeps small] 
\label{lem:caseII_noreg}  For any $c \geq \max(32\rho,288/(\eta d))$ and for a sufficiently large $n$, there exists $T \geq 1$ (it can be the same stated in the previous lemma) such that, for each $t$ in the range  $[T , \ldots,  3 \log n]$,
\begin{align}
\label{eq:caseII_noreg}
    \mathbf{Pr} \bigr(K_t \leq \delta_t' \mid K_1 \leq \gamma_1',\dots,K_{T-1}\leq \gamma_{T-1}',K_T \leq \delta_T',\dots,K_{t-1}\leq \delta_{t-1}' \bigr) \geq 1-\frac{1}{n^3},
\end{align} 
where $\gamma_t'$ is defined in \eqref{def:gamma1} and $\delta_t'$ is defined by the recurrence
\begin{equation}
    \delta_t'=\frac{1}{4}+\frac{24 t \log n}{cd \Deltamin(\sC)}, \text{ for } t \geq T\,.
    \label{def:delta1}
\end{equation}
\end{lemma}
\begin{proof}
As in the proof of the previous lemma, we take $T$ as the first integer such that
\begin{equation}
    \label{eq:cond_T}
    d \Deltamax(\sS) \prod_{j=0}^{T-1}\gamma_j' \leq 12 \log n \,.
\end{equation}
Observe first that, for each $t \leq 3 \log n$ and for $c \geq 288/(d \eta)$, we have that $\delta_t' \leq 1/2$. So, for each $t$ s.t. $T \leq t \leq 3 \log n$, \eqref{eq:cond_T} and Lemma \ref{lem:chernoff_noreg} imply that
\begin{align}
    &\Expcc{r_t(N(v)) \mid K_1 \leq \gamma_1',\dots,K_T \leq \gamma_T',\dots,K_{t-1}\leq \delta_{t-1}'} 
    \leq d \Delta_v \frac{\Deltamax(\sS)}{\Deltamin(\sC)}\prod_{j=0}^{T-1}\gamma_j' \prod_{i=T}^{t-1}\delta_i' \\&\leq d \Delta_v \frac{\Deltamax(\sC)}{\Deltamin(\sS)} \prod_{j=0}^{T-1}\gamma_j' \leq 12 \log n\,.
\end{align}
Taking $\mu=12 \log n$, for Lemma \ref{lem:chernoff_noreg}, \eqref{eq:cond_T} and, by a union bound over all the clients $v \in \sC$, 
\begin{align}
    & \mathbf{Pr} \bigr( \bigcap_{v \in \sC}\bigr\{r_t(N(v)) \leq \frac{12\Delta_v}{\Deltamin(\sC)} \log n \bigr\}  \mid K_1,\dots,\gamma_1',K_{T-1}\leq \gamma_{T-1}',\dots,K_{T}\leq \delta_{T}',\dots,K_{t-1}\leq \delta_{t-1}' \bigr) \\ &\geq 1-\frac{1}{n^3}\,.
     \label{eq:bound_r_t(N(v))_hp_noreg}
\end{align}
\eqref{eq:caseII_noreg} follows from \eqref{eq:bound_r_t(N(v))_hp_noreg}, from the definition of $\delta_t'$ in \eqref{def:delta1} and from \eqref{def:kt_noreg_prop}.
\end{proof}

Lemma \ref{lem:caseI_noreg} and \ref{lem:caseII_noreg} imply Lemma \ref{lem:bound_burned_alm_reg}. Indeed, for the chain rule, taking $T'= \lfloor 3 \log n \rfloor$ and $c \geq \max(32 \rho, 288/(\eta d))$ we get
\begin{align}
    &\Prc{\cap_{t=1}^{T-1}\{K_t \leq \gamma_t'\}\bigcap \cap_{t=T}^{T'} \{K_t \leq \delta_t'\}}= 
    \left(1-\frac{1}{n^3}\right)^{T'} \geq 1-T'\frac{1}{n^3} \geq 1-\frac{1}{n^2}\,.
    \label{high_probability_noreg}
\end{align}
where in the first inequality of \eqref{high_probability_noreg} we used the chain rule, Lemma \ref{lem:caseI_noreg} and \ref{lem:caseII_noreg} while the second last inequality of \eqref{high_probability_noreg} follows from the binomial inequality, i.e., for each $x \geq -1$ and for each $m \in \mathbb{N}$,$(1+x)^m \geq 1+mx$. 
Concluding, we have shown that $K_t \leq \gamma_t'$ for each $t \leq T$, and $K_t \leq \delta_t'$ for all $t$ such that $T \leq t \leq 3 \log n$, with probability at least $1-1/n^2$. So, since $S_t \leq K_t$ and $c \geq \max(32\rho,288/(\eta d)$,   from \eqref{def:delta1} and Lemma \ref{claim:recurrence}, with probability at least $1-1/n^2$, 
$S_t \leq 1/2$ for all $t$ such that $t \leq 3 \log n$.
\end{proof}


\end{document}